\documentclass[11pt]{article}
\usepackage{amsmath,amssymb,amsfonts}
\usepackage{graphicx}
\usepackage{fullpage}
\usepackage[backref]{hyperref}
\usepackage{color}
\usepackage{wrapfig}
\usepackage{tikz}
\usepackage{setspace}
\usepackage[framemethod=tikz]{mdframed}
\usepackage{xspace}
\usepackage{multirow}

\newenvironment{proof}{\noindent{\bf Proof : \ }}{\hfill$\Box$\par\medskip}

\usetikzlibrary{decorations.pathreplacing}

\newtheorem{theorem}{Theorem}[section]
\newtheorem{corollary}[theorem]{Corollary}
\newtheorem{lemma}[theorem]{Lemma}

\newtheorem{definition}[theorem]{Definition}

\newtheorem{claim}[theorem]{Claim}

\newtheorem{observation}[theorem]{Observation}

\newenvironment{proofof}[1]{\begin{trivlist} \item {\bf Proof
#1:~~}}
  {\qed\end{trivlist}}
\renewenvironment{proofof}[1]{\par\medskip\noindent{\bf Proof of #1: \ }}{\hfill$\Box$\par\medskip}
\newcommand{\namedref}[2]{\hyperref[#2]{#1~\ref*{#2}}}
\newcommand{\thmlab}[1]{\label{thm:#1}}
\newcommand{\thmref}[1]{\namedref{Theorem}{thm:#1}}
\newcommand{\lemlab}[1]{\label{lemm:#1}}
\newcommand{\lemref}[1]{\namedref{Lemma}{lemm:#1}}
\newcommand{\obslab}[1]{\label{obs:#1}}
\newcommand{\obsref}[1]{\namedref{Observation}{obs:#1}}
\newcommand{\figlab}[1]{\label{fig:#1}}
\newcommand{\figref}[1]{\namedref{Figure}{fig:#1}}
\newcommand{\seclab}[1]{\label{sec:#1}}
\newcommand{\secref}[1]{\namedref{Section}{sec:#1}}

\newcommand{\tablab}[1]{\label{tab:#1}}
\newcommand{\tabref}[1]{\namedref{Table}{tab:#1}}

\newenvironment{remindertheorem}[1]{\medskip \noindent {\bf Reminder of  #1.  }\em}{}

\def \HAM    {\mdef{\mathsf{HAM}}}

\newcommand{\COMMENTED}[1]{{}}

\newcommand{\PPr}[1]{\ensuremath{\mathbf{Pr}\Big[#1\Big]}}

\newcommand{\bigO}[1]{\ensuremath{\mathcal{O}\left(#1\right)}}

\newcommand{\eps}{\epsilon}
\newcommand{\level}[1]{\ensuremath{\mathit{level}\left(#1\right)}}
\newcommand{\fingerprint}[1]{\ensuremath{\mathit{fingerprints}\left(#1\right)}}

\def \check {\mdef{\mathsf{NearPalindrome}}}
\def \recover {\mdef{\mathsf{Recover}}}

\newcommand{\mdef}[1]{{\ensuremath{#1}}\xspace}  
\newcommand{\myfunc}[1]{\mdef{\mathsf{#1}}}      

\DeclareMathOperator*{\polylog}{polylog}

\newcommand{\superscript}[1]{\ensuremath{^{\mbox{\tiny{\textit{#1}}}}}\xspace}
\def \th {\superscript{th}}     
\def \etal{{\it et~al.}}

\def \polylog  {\mdef{\myfunc{polylog}}}             
\newcommand{\flr}[1]{\mdef{\left\lfloor#1\right\rfloor}}              
\newcommand{\ceil}[1]{\mdef{\left\lceil#1\right\rceil}}               

\newcommand{\ignore}[1]{}

\newif\ifnotes\notestrue 
\ifnotes
\newcommand{\elena}[1]{\textcolor{red}{{\bf (Elena:} {#1}{\bf ) }} \marginpar{\tiny\bf
             \begin{minipage}[t]{0.5in}
               \raggedright E:
                \end{minipage}}}
\newcommand{\erfan}[1]{\textcolor{blue}{{\bf (Erfan:} {#1}{\bf ) }} \marginpar{\tiny\bf
             \begin{minipage}[t]{0.5in}
               \raggedright E:
                \end{minipage}}}
\newcommand{\samson}[1]{\textcolor{green}{{\bf (Samson:} {#1}{\bf ) }} \marginpar{\tiny\bf
             \begin{minipage}[t]{0.5in}
               \raggedright S:
            \end{minipage}}}            							
\else
\newcommand{\elena}[1]{}
\newcommand{\erfan}[1]{}
\newcommand{\samson}[1]{}
\fi

\hypersetup{
     colorlinks   = true,
     citecolor    = blue,
		 linkcolor		= red
}

\title{Streaming for Aibohphobes: Longest Palindrome with Mismatches}

\author{
Elena Grigorescu\thanks{Department of Computer Science, Purdue University, West Lafayette, IN. 
Email: {\tt elena-g@purdue.edu}. Research supported in part by NSF CCF-1649515.}
\and
Erfan Sadeqi Azer\thanks{School of Informatics and Computing, Indiana University Bloomington, Bloomington, IN.
Email: {\tt esadeqia@indiana.edu}.}
\and
Samson Zhou\thanks{Department of Computer Science, Purdue University, West Lafayette, IN. 
Email: {\tt samsonzhou@gmail.com}. Research supported by NSF CCF-1649515.}
}

\begin{document}
\maketitle
\begin{abstract}
A palindrome is a string that reads the same as its reverse, such as ``aibohphobia'' (fear of palindromes).
Given an integer $d>0$, a {\em $d$-near-palindrome} is a string of Hamming distance at most $d$ from its reverse.
 
We study the natural problem of identifying a longest $d$-near-palindrome in data streams. The problem is relevant to the analysis of DNA databases, and to the task of repairing recursive structures in documents such as XML and JSON.

We present an algorithm that returns a $d$-near-palindrome whose length is within a multiplicative $(1+\eps)$-factor of the longest $d$-near-palindrome. 
Our algorithm also returns the set of mismatched indices of the $d$-near-palindrome,  using $\bigO{\frac{d\log^7 n}{\eps\log(1+\eps)}}$ bits of space, and $\bigO{\frac{d\log^6 n}{\eps\log(1+\eps)}}$ update time per arriving symbol. 
We show that $\Omega(d\log n)$ space is necessary for estimating the length of longest $d$-near-palindromes  with high probability.

We further obtain an additive-error approximation algorithm and a comparable lower bound, as well as an {\em exact} two-pass algorithm that  solves the longest $d$-near-palindrome problem using $\bigO{d^2\sqrt{n}\log^6 n}$  bits of space.
\end{abstract}
\section{Introduction}
A palindrome is a string that  reads the same as its reverse, such as the common construct  ``racecar'', or the deliberate construct ``aibohphobia''.  Given a metric and an integer $d>0$, we say that a string is a {\em  $d$-near-palindrome} if it is at distance at most $d$ from its reverse.
In this paper, we study the problem of identifying the longest $d$-near-palindrome substring in the {\em streaming} model,  under the Hamming distance.  
In the streaming model,  the input data arrives one symbol at a time, and we are allowed to perform computation using only a small amount of working memory. 
Specifically, our goal is to approximate the length of a longest near-palindrome in a string of length $n$, using only $o(n)$ space. 
A related question regarding approximating the length of a longest palindrome in RNA sequences under removal of elements was explicitly asked at the Bertinoro Workshop on Sublinear Algorithms 2014 \cite{SublinearOpen61}.

Finding near-palindromes is widely motivated in string processing of databases relevant to bioinformatics.
Specifically, since the development of the Human Genome Project, advances in biological algorithms have quickened the sequencing for genes and proteins, leading to increasingly large databases of strings representing both nucleic acids for DNA or RNA, and amino acids for proteins. 
Tools to analyze these sequences, such as the basic local alignment search tool (BLAST) \cite{Altschul90}, often require the removal of ``low-complexity'' regions (long repetitive or palindromic structures). 
However, these long sequences frequently contain small perturbations through mutation or some other form of corruption (including human error),  so that identifying ``near''-palindromes under either Hamming distance or edit distance is important for preprocessing sequences before applying  heuristic tools.
In particular, the streaming model is relevant to contemporary data-sequencing technologies for near-palindromes, as further discussed in \cite{ContiCH04, HerbordtMSGC07}.
\subsection*{Our contributions}
We initiate the study of  finding near-palindromes in the streaming model, and provide several algorithms 
for the longest near-palindrome substring. 

Given a stream $S$ of length $n$ and integer $d=o(\sqrt{n})$, let $\ell_{max}$ be the length of a longest $d$-near-palindrome substring in $S$. 

\newcommand{\thmham}{There exists a one-pass streaming algorithm that returns a $d$-near-palindrome of length at least $\frac{1}{1+\eps}\cdot {\ell_{max}}$, with probability $1-\frac{1}{n}$. 
The algorithm uses $\bigO{\frac{d\log^7 n}{\eps\log(1+\eps)}}$ bits of space and update time $\bigO{\frac{d\log^6 n}{\eps\log(1+\eps)}}$ per arriving symbol.
}
\begin{theorem} \thmlab{thm:ham}
\thmham
\end{theorem}

\newcommand{\thmadd}{There exists a one-pass streaming algorithm that returns a $d$-near-palindrome of length at least $\ell_{max}-E$, with probability $1-\frac{1}{n}$. 
The algorithm uses $\bigO{\frac{dn\log^6 n}{E}}$ bits of space and update time $\bigO{\frac{dn\log^5 n}{E}}$ per arriving symbol.
}
\begin{theorem} \thmlab{thm:add}
\thmadd
\end{theorem}
If two passes over the stream are allowed,  one can find an {\em exact} longest $d$-near-palindrome.

\newcommand{\thmexact}{There exists a two-pass streaming algorithm that returns a $d$-near-palindrome of length $\ell_{max}$, with probability $1-\frac{1}{n}$. 
It uses $\bigO{d^2\sqrt{n}\log^6 n}$ bits of space and $\bigO{d^2\sqrt{n}\log^5 n}$ update time per arriving symbol.
}
\begin{theorem} \thmlab{thm:exact}
\thmexact
\end{theorem}
We complement our results with lower bounds for randomized algorithms.
 
\newcommand{\lowerbounds}{Let $d=o(\sqrt{n})$. 
Any randomized streaming algorithm that returns an estimate $\hat{\ell}$ of the length of the longest $d$-near-palindrome, where $\hat{\ell}\le\ell_{max}\le(1+\eps)\hat{\ell}$, with probability at least $1-\frac{1}{n}$, must use $\Omega\left(d\log n\right)$ bits of space.}
\begin{theorem} \thmlab{thm:lowerbounds}
\lowerbounds
\end{theorem}

\newcommand{\lowerboundsb}{Let $d=o(\sqrt{n})$ and $E>d$ be an integer. 
Any randomized streaming algorithm that returns an estimate $\hat{\ell}$ of the length of the longest $d$-near-palindrome, where $\hat{\ell}\le\ell_{max}\le\hat{\ell}+E$, with probability at least $1-\frac{1}{n}$, must use $\Omega\left(\frac{dn}{E}\right)$ bits of space.}
\begin{theorem} \thmlab{thm:lowerboundsb}
\lowerboundsb
\end{theorem}
A summary of our results and comparison with related work appears in \tabref{table:results}.
\subsection*{Background and Related Work}
Our techniques extend previous work on the Longest Palindromic Substring Problem, the Pattern Matching Problem, and the $d$-Mismatch Problem in the streaming model.

In the \emph{Longest Palindromic Substring Problem}, the goal is to output a longest palindromic substring of an input of length $n$, while minimizing the computation space. 
Manacher \cite{Manacher75} introduces a linear-time online algorithm that reports whether all symbols seen at the time of query form a palindrome. 
Berenbrink \etal\  \cite{BerenbrinkEMA14} achieve $\bigO{\frac{\log^2 n}{\eps\log(1+\eps)}}$ space for multiplicative error $(1+\eps)$, and show a space lower bound for algorithms with additive error. 
Gawrychowski \etal\  \cite{GawrychowskiMSU16} recently generalize the aforementioned lower bounds for additive error, and also produce a space lower bound of $\Omega\left(\frac{\log n}{\log(1+\eps)}\right)$ for algorithms with multiplicative error $(1+\eps)$.

In the \emph{Pattern Matching Problem}, one is given a pattern of length $m$ and the goal is to output all occurrences of the pattern in the input string, while again minimizing space or update time.
In order to achieve space sublinear in the size of the input, many pattern matching streaming algorithms use Karp-Rabin fingerprints \cite{KarpR87}. 
Porat and Porat \cite{PoratP09} present a randomized algorithm for exact pattern matching using $\bigO{\log m}$ space and $\bigO{\log m}$ update time, which Breslauer and Galil \cite{BreslauerG14} further improve to constant update time. 
For a more comprehensive survey on pattern matching, see \cite{ApostolicoG:1997}. 

In the related \emph{$d$-Mismatch Problem}, one is given a pattern of length $m$ and the goal is to find all substrings of the input that are at most Hamming distance $d$ from the pattern.
A line of exciting work (e.g., \cite{AmirLP04, PoratP09, CliffordEPP11, AndoniGMP13}) culminates in a recent algorithm by Clifford \etal\ \cite{CliffordFPSS16} that uses $\bigO{d^2\,\polylog\,m}$ space and $\bigO{\sqrt{d}\log d+\polylog\,m}$ update time per arriving symbol.

For several other metrics, Clifford \etal\  \cite{CliffordJPS13} show that linear space is necessary for algorithms identifying substrings with distance at most $d$ from a given pattern. 
Similarly, Andoni \etal\ \cite{AndoniGMP13} prove that any sketch estimating the edit distance between two strings requires space almost linear in the inputs. 
For time bounds, Backurs and Indyk \cite{BackursI15} show that the strong Exponential Time Hypothesis implies the general edit distance problem cannot be solved in time better than $n^{2-\eps}$.
On the positive side, Chakraborty \etal\ \cite{ChakrabortyGK16} give a low distortion embedding from edit distance to Hamming distance,  and Belazzougui and Zhang \cite{BelazzouguiZ16} provide the first streaming algorithm for computing edit distance using $\bigO{d^8\log^5 n}$ space, given the promise that the edit distance is at most $d$.

\begin{table*}
\centering
\resizebox{14cm}{!}{
\begin{tabular}[htb]{|c|c|c|c|c|}\hline
&\multicolumn{2}{|c|}{Space for Algorithms} & \multicolumn{2}{|c|}{Lower Bounds}\\\hline
Model & $d$-Near-Palindrome & Palindrome & $d$-Near-Palindrome & Palindrome \\\hline

$1$-Pass, Multiplicative $(1+\eps)$ & $\bigO{\frac{d\log^7 n}{\eps\log(1+\eps)}}$ & $\bigO{\frac{\log^2 n}{\eps\log(1+\eps)}}$\cite{BerenbrinkEMA14} & $\Omega(d\log n)$ & $\Omega\left(\frac{\log n}{\log(1+\eps)}\right)$\cite{GawrychowskiMSU16} \\\hline

$1$-Pass, Additive $E$ & $\bigO{\frac{dn\log^6 n}{E}}$ & $\bigO{\frac{n\log n}{E}}$\cite{BerenbrinkEMA14} & $\Omega\left(\frac{dn}{E}\right)$ & $\Omega\left(\frac{n}{E}\right)$\cite{BerenbrinkEMA14} \\\hline

$2$-Pass, Exact & $\bigO{d^2\sqrt{n}\log^6 n}$  & $\bigO{\sqrt{n}\log n}$\cite{BerenbrinkEMA14} & - & - \\\hline
\end{tabular}
}
\caption{Summary of our results and comparison to related work}
\tablab{table:results}
\end{table*}
\section{Preliminaries}\seclab{sec:prelim}
We denote by $[n]$ the set $\{1, 2, \ldots, n\}$. 
We assume an input stream of length $n$ over alphabet $\Sigma$. 
Given a string $S[1,\dots,n]$, we denote its length by $|S|$, its $i\th$ character by $S[i]$ or $S_i$, and the substring between locations $i$ and $j$ (inclusive) by $S[i,j]$.  

The Hamming distance between $S$ and $T$, denoted $\HAM(S,T)$, is the number of indices whose symbols do not match: $\HAM(S,T)=\Big|\{ i\mid S[i]\ne T[i]\}\Big|$.
We denote the concatenation of $S$ and $T$ by $S\circ T$. 
Each index $i$ such that $S[i]\neq S[n-i+1]$ is a \emph{mismatch}.
We say $S$ is a $d$-\emph{near-palindrome} if $\HAM(S,S^R)\le d$. 
Without loss of generality, our algorithms assume the lengths of $d$-near-palindromes are even, since for any odd length $d$-near-palindrome, we may apply the algorithm to $S[1]S[1]S[2]S[2]\cdots S[n]S[n]$ instead of $S[1,n]$.

\begin{definition}[Karp-Rabin Fingerprint]
For a string $S$, prime $P$ and integer $B$ with $1\le B<P$, the Karp-Rabin forward and reverse fingerprints \cite{KarpR87} are defined as follows: 
\begin{align*}
\phi^F(S)=\left(\sum_{x=1}^{|S|}S[x]\cdot B^x\right)\bmod{P},\qquad
\phi^R(S)=\left(\sum_{x=1}^{|S|}S[x]\cdot B^{-x}\right)\bmod{P}.
\end{align*}
\end{definition}

\noindent  Karp-Rabin Fingerprints have the following easily verifiable properties:
\begin{align*}
1. \quad & \phi^R(S)\cdot B^{|S|+1}=\phi(S^R)\bmod{P} & \text{(reversal)} \\
2. \quad & \phi^F(S[x,y])=B^{1-x}(\phi^F(S[1,y])-\phi^F(S[1,x-1]))\bmod{P} & \text{(sliding)} \\
3. \quad & \phi^R(S[x,y])=B^{x-1}(\phi^R(S[1,y])-\phi^R(S[1,x-1]))\bmod{P} & \text{(sliding)}
\end{align*}
We use Karp-Rabin Fingerprints for certain subpatterns of $S$, as in \cite{CliffordFPSS16}. 
For a string $S$ and integers $a\le b$, define the \emph{first-level subpattern} $S_{a,b}$ to be the subsequence $S[a]S[a+b]S[a+2b]\ldots$. 
In this case, define $S^R_{a,b}=(S_{a,b})^R$ (as opposed to $(S^R)_{a,b}$). 
Similarly, define $S_{a,b}[x,y]=S_{a,b}\cap S[x,y]$ (as opposed to $(S[x,y])_{a,b}$).
Then for $1\le a\le b$, define the fingerprints for $S_{a,b}$ and its reverse:
\begin{align*}
\phi^F_{a,b}(S)=\phi^F(S_{a,b})&=\left(\sum_{x\equiv a\bmod{b}}S[x]\cdot B^{\ceil{x/b}}\right)\bmod{P}\\
\phi^R_{a,b}(S)=\phi^R(S_{a,b})&=\left(\sum_{x\equiv a\bmod{b}}S[x]\cdot B^{-\ceil{x/b}}\right)\bmod{P}
\end{align*}
For an example, see \figref{fig:subpattern}.

\begin{figure*}[htb]
\centering
\begin{tikzpicture}[scale=0.6]
\draw (1cm,0cm) -- (10cm,0cm);
\draw (1cm,1.2cm) -- (10cm,1.2cm);
\filldraw[thick, top color=white,bottom color=red!50!] (1cm,0cm) rectangle+(0.6cm,1.2cm);
\filldraw[thick, top color=white,bottom color=blue!50!] (1.6cm,0cm) rectangle+(0.6cm,1.2cm);
\filldraw[thick, top color=white,bottom color=green!50!] (2.2cm,0cm) rectangle+(0.6cm,1.2cm);
\filldraw[thick, top color=white,bottom color=red!50!] (2.8cm,0cm) rectangle+(0.6cm,1.2cm);
\filldraw[thick, top color=white,bottom color=blue!50!] (3.4cm,0cm) rectangle+(0.6cm,1.2cm);
\filldraw[thick, top color=white,bottom color=green!50!] (4cm,0cm) rectangle+(0.6cm,1.2cm);
\filldraw[thick, top color=white,bottom color=red!50!] (4.6cm,0cm) rectangle+(0.6cm,1.2cm);
\filldraw[thick, top color=white,bottom color=blue!50!] (5.2cm,0cm) rectangle+(0.6cm,1.2cm);
\filldraw[thick, top color=white,bottom color=green!50!] (5.8cm,0cm) rectangle+(0.6cm,1.2cm);
\filldraw[thick, top color=white,bottom color=red!50!] (6.4cm,0cm) rectangle+(0.6cm,1.2cm);
\filldraw[thick, top color=white,bottom color=blue!50!] (7cm,0cm) rectangle+(0.6cm,1.2cm);
\filldraw[thick, top color=white,bottom color=green!50!] (7.6cm,0cm) rectangle+(0.6cm,1.2cm);
\filldraw[thick, top color=white,bottom color=red!50!] (8.2cm,0cm) rectangle+(0.6cm,1.2cm);
\filldraw[thick, top color=white,bottom color=blue!50!] (8.8cm,0cm) rectangle+(0.6cm,1.2cm);
\filldraw[thick, top color=white,bottom color=green!50!] (9.4cm,0cm) rectangle+(0.6cm,1.2cm);
\node at (0.4cm, 0.6cm) {$S$:};

\filldraw[thick, top color=white,bottom color=red!50!] (-2cm,-3cm) rectangle+(0.6cm,1.2cm);
\filldraw[thick, top color=white,bottom color=red!50!] (-1.4cm,-3cm) rectangle+(0.6cm,1.2cm);
\filldraw[thick, top color=white,bottom color=red!50!] (-0.8cm,-3cm) rectangle+(0.6cm,1.2cm);
\filldraw[thick, top color=white,bottom color=red!50!] (-0.2cm,-3cm) rectangle+(0.6cm,1.2cm);
\filldraw[thick, top color=white,bottom color=red!50!] (0.4cm,-3cm) rectangle+(0.6cm,1.2cm);
\draw (1.3cm, -0.1cm) --(-1.7cm, -1.7cm);
\draw (3.1cm, -0.1cm) --(-1.1cm, -1.7cm);
\draw (4.9cm, -0.1cm) --(-0.5cm, -1.7cm);
\draw (6.7cm, -0.1cm) --(0.1cm, -1.7cm);
\draw (8.3cm, -0.1cm) --(0.7cm, -1.7cm);
\node at (-3.2cm, -2.4cm) {$\phi^F_{1,3}(S)$:};

\filldraw[thick, top color=white,bottom color=blue!50!] (3.6cm,-3cm) rectangle+(0.6cm,1.2cm);
\filldraw[thick, top color=white,bottom color=blue!50!] (4.2cm,-3cm) rectangle+(0.6cm,1.2cm);
\filldraw[thick, top color=white,bottom color=blue!50!] (4.8cm,-3cm) rectangle+(0.6cm,1.2cm);
\filldraw[thick, top color=white,bottom color=blue!50!] (5.4cm,-3cm) rectangle+(0.6cm,1.2cm);
\filldraw[thick, top color=white,bottom color=blue!50!] (6cm,-3cm) rectangle+(0.6cm,1.2cm);
\draw (1.9cm, -0.1cm) --(3.9cm, -1.7cm);
\draw (3.7cm, -0.1cm) --(4.5cm, -1.7cm);
\draw (5.5cm, -0.1cm) --(5.1cm, -1.7cm);
\draw (7.3cm, -0.1cm) --(5.7cm, -1.7cm);
\draw (8.9cm, -0.1cm) --(6.3cm, -1.7cm);
\node at (2.4cm, -2.4cm) {$\phi^F_{2,3}(S)$:};

\filldraw[thick, top color=white,bottom color=green!50!] (9.2cm,-3cm) rectangle+(0.6cm,1.2cm);
\filldraw[thick, top color=white,bottom color=green!50!] (9.8cm,-3cm) rectangle+(0.6cm,1.2cm);
\filldraw[thick, top color=white,bottom color=green!50!] (10.4cm,-3cm) rectangle+(0.6cm,1.2cm);
\filldraw[thick, top color=white,bottom color=green!50!] (11cm,-3cm) rectangle+(0.6cm,1.2cm);
\filldraw[thick, top color=white,bottom color=green!50!] (11.6cm,-3cm) rectangle+(0.6cm,1.2cm);
\draw (2.5cm, -0.1cm) --(9.5cm, -1.7cm);
\draw (4.3cm, -0.1cm) --(10.1cm, -1.7cm);
\draw (6.1cm, -0.1cm) --(10.7cm, -1.7cm);
\draw (7.9cm, -0.1cm) --(11.3cm, -1.7cm);
\draw (9.5cm, -0.1cm) --(11.9cm, -1.7cm);
\node at (8.0cm, -2.4cm) {$\phi^F_{3,3}(S)$:};
\end{tikzpicture}
\caption{Karp-Rabin Fingerprints for first-level subpattern.}\figlab{fig:subpattern}
\end{figure*}
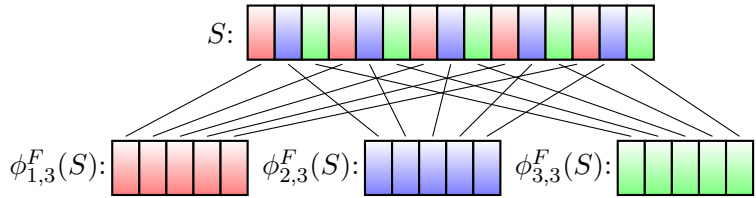

Given a first-level subpattern $T=S_{a,b}=S[a]S[a+b]S[a+2b]\ldots$ and integers $r\le s$, define the \emph{second-level subpattern} $T_{r,s}=T[r]T[r+s]T[r+2s]\ldots$. 
Observe that $T_{r,s}=S_{a+rb,sb}$ and thus, second-level subpatterns are simply more refined first-level subpatterns. 
For an example, see \figref{fig:2nd:subpattern}.

Observe the following properties of fingerprints on first-level and second-level subpatterns: 
\begin{align*}
1. \qquad & \phi_{a,b}^R(S)\cdot B^{|S|+1}=\phi_{|S|-a+1,b}(S^R)\bmod{P} & \text{(reversal)} \\
2. \qquad & \phi_{a,b}^F(S[x,y])=B^{\ceil{(1-x)/b}}(\phi_{a,b}^F(S[1,y])-\phi_{a,b}^F(S[1,x-1]))\bmod{P} & \text{(sliding)} \\
3. \qquad & \phi_{a,b}^R(S[x,y])=B^{\ceil{(x-1)/b}}(\phi_{a,b}^R(S[1,y])-\phi_{a,b}^R(S[1,x-1]))\bmod{P} & \text{(sliding)}
\end{align*}

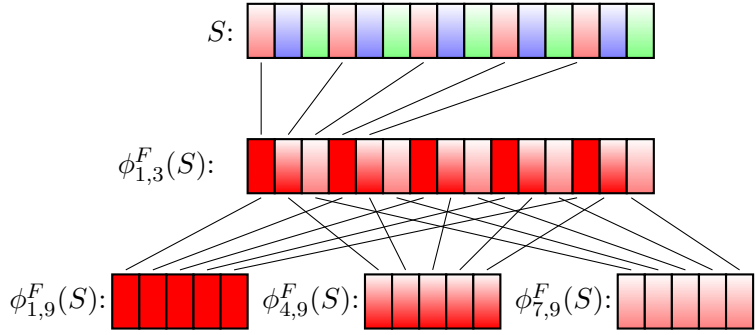
\begin{figure*}[htb]
\centering
\begin{tikzpicture}[scale=0.6]
\draw (1cm,0cm) -- (10cm,0cm);
\draw (1cm,1.2cm) -- (10cm,1.2cm);
\filldraw[thick, top color=white,bottom color=red!50!] (1cm,0cm) rectangle+(0.6cm,1.2cm);
\filldraw[thick, top color=white,bottom color=blue!50!] (1.6cm,0cm) rectangle+(0.6cm,1.2cm);
\filldraw[thick, top color=white,bottom color=green!50!] (2.2cm,0cm) rectangle+(0.6cm,1.2cm);
\filldraw[thick, top color=white,bottom color=red!50!] (2.8cm,0cm) rectangle+(0.6cm,1.2cm);
\filldraw[thick, top color=white,bottom color=blue!50!] (3.4cm,0cm) rectangle+(0.6cm,1.2cm);
\filldraw[thick, top color=white,bottom color=green!50!] (4cm,0cm) rectangle+(0.6cm,1.2cm);
\filldraw[thick, top color=white,bottom color=red!50!] (4.6cm,0cm) rectangle+(0.6cm,1.2cm);
\filldraw[thick, top color=white,bottom color=blue!50!] (5.2cm,0cm) rectangle+(0.6cm,1.2cm);
\filldraw[thick, top color=white,bottom color=green!50!] (5.8cm,0cm) rectangle+(0.6cm,1.2cm);
\filldraw[thick, top color=white,bottom color=red!50!] (6.4cm,0cm) rectangle+(0.6cm,1.2cm);
\filldraw[thick, top color=white,bottom color=blue!50!] (7cm,0cm) rectangle+(0.6cm,1.2cm);
\filldraw[thick, top color=white,bottom color=green!50!] (7.6cm,0cm) rectangle+(0.6cm,1.2cm);
\filldraw[thick, top color=white,bottom color=red!50!] (8.2cm,0cm) rectangle+(0.6cm,1.2cm);
\filldraw[thick, top color=white,bottom color=blue!50!] (8.8cm,0cm) rectangle+(0.6cm,1.2cm);
\filldraw[thick, top color=white,bottom color=green!50!] (9.4cm,0cm) rectangle+(0.6cm,1.2cm);
\node at (0.4cm, 0.6cm) {$S$:};

\filldraw[thick, top color=red,bottom color=red] (1cm,-3cm) rectangle+(0.6cm,1.2cm);
\filldraw[thick, top color=white,bottom color=red!100!] (1.6cm,-3cm) rectangle+(0.6cm,1.2cm);
\filldraw[thick, top color=white,bottom color=red!50!] (2.2cm,-3cm) rectangle+(0.6cm,1.2cm);
\filldraw[thick, top color=red,bottom color=red] (2.8cm,-3cm) rectangle+(0.6cm,1.2cm);
\filldraw[thick, top color=white,bottom color=red!100!] (3.4cm,-3cm) rectangle+(0.6cm,1.2cm);
\filldraw[thick, top color=white,bottom color=red!50!] (4cm,-3cm) rectangle+(0.6cm,1.2cm);
\filldraw[thick, top color=red,bottom color=red] (4.6cm,-3cm) rectangle+(0.6cm,1.2cm);
\filldraw[thick, top color=white,bottom color=red!100!] (5.2cm,-3cm) rectangle+(0.6cm,1.2cm);
\filldraw[thick, top color=white,bottom color=red!50!] (5.8cm,-3cm) rectangle+(0.6cm,1.2cm);
\filldraw[thick, top color=red,bottom color=red] (6.4cm,-3cm) rectangle+(0.6cm,1.2cm);
\filldraw[thick, top color=white,bottom color=red!100!] (7cm,-3cm) rectangle+(0.6cm,1.2cm);
\filldraw[thick, top color=white,bottom color=red!50!] (7.6cm,-3cm) rectangle+(0.6cm,1.2cm);
\filldraw[thick, top color=red,bottom color=red] (8.2cm,-3cm) rectangle+(0.6cm,1.2cm);
\filldraw[thick, top color=white,bottom color=red!100!] (8.8cm,-3cm) rectangle+(0.6cm,1.2cm);
\filldraw[thick, top color=white,bottom color=red!50!] (9.4cm,-3cm) rectangle+(0.6cm,1.2cm);
\draw (1.3cm, -0.1cm) --(1.3cm, -1.7cm);
\draw (3.1cm, -0.1cm) --(1.9cm, -1.7cm);
\draw (4.9cm, -0.1cm) --(2.5cm, -1.7cm);
\draw (6.7cm, -0.1cm) --(3.1cm, -1.7cm);
\draw (8.3cm, -0.1cm) --(3.7cm, -1.7cm);
\node at (-0.8cm, -2.4cm) {$\phi^F_{1,3}(S)$:};

\filldraw[thick, top color=red,bottom color=red] (-2cm,-6cm) rectangle+(0.6cm,1.2cm);
\filldraw[thick, top color=red,bottom color=red] (-1.4cm,-6cm) rectangle+(0.6cm,1.2cm);
\filldraw[thick, top color=red,bottom color=red] (-0.8cm,-6cm) rectangle+(0.6cm,1.2cm);
\filldraw[thick, top color=red,bottom color=red] (-0.2cm,-6cm) rectangle+(0.6cm,1.2cm);
\filldraw[thick, top color=red,bottom color=red] (0.4cm,-6cm) rectangle+(0.6cm,1.2cm);
\draw (1.3cm, -3.1cm) --(-1.7cm, -4.7cm);
\draw (3.1cm, -3.1cm) --(-1.1cm, -4.7cm);
\draw (4.9cm, -3.1cm) --(-0.5cm, -4.7cm);
\draw (6.7cm, -3.1cm) --(0.1cm, -4.7cm);
\draw (8.3cm, -3.1cm) --(0.7cm, -4.7cm);
\node at (-3.2cm, -5.4cm) {$\phi^F_{1,9}(S)$:};

\filldraw[thick, top color=white,bottom color=red!100!] (3.6cm,-6cm) rectangle+(0.6cm,1.2cm);
\filldraw[thick, top color=white,bottom color=red!100!] (4.2cm,-6cm) rectangle+(0.6cm,1.2cm);
\filldraw[thick, top color=white,bottom color=red!100!] (4.8cm,-6cm) rectangle+(0.6cm,1.2cm);
\filldraw[thick, top color=white,bottom color=red!100!] (5.4cm,-6cm) rectangle+(0.6cm,1.2cm);
\filldraw[thick, top color=white,bottom color=red!100!] (6cm,-6cm) rectangle+(0.6cm,1.2cm);
\draw (1.9cm, -3.1cm) --(3.9cm, -4.7cm);
\draw (3.7cm, -3.1cm) --(4.5cm, -4.7cm);
\draw (5.5cm, -3.1cm) --(5.1cm, -4.7cm);
\draw (7.3cm, -3.1cm) --(5.7cm, -4.7cm);
\draw (8.9cm, -3.1cm) --(6.3cm, -4.7cm);
\node at (2.4cm, -5.4cm) {$\phi^F_{4,9}(S)$:};

\filldraw[thick, top color=white,bottom color=red!50!] (9.2cm,-6cm) rectangle+(0.6cm,1.2cm);
\filldraw[thick, top color=white,bottom color=red!50!] (9.8cm,-6cm) rectangle+(0.6cm,1.2cm);
\filldraw[thick, top color=white,bottom color=red!50!] (10.4cm,-6cm) rectangle+(0.6cm,1.2cm);
\filldraw[thick, top color=white,bottom color=red!50!] (11cm,-6cm) rectangle+(0.6cm,1.2cm);
\filldraw[thick, top color=white,bottom color=red!50!] (11.6cm,-6cm) rectangle+(0.6cm,1.2cm);
\draw (2.5cm, -3.1cm) --(9.5cm, -4.7cm);
\draw (4.3cm, -3.1cm) --(10.1cm, -4.7cm);
\draw (6.1cm, -3.1cm) --(10.7cm, -4.7cm);
\draw (7.9cm, -3.1cm) --(11.3cm, -4.7cm);
\draw (9.5cm, -3.1cm) --(11.9cm, -4.7cm);
\node at (8.0cm, -5.4cm) {$\phi^F_{7,9}(S)$:};

\end{tikzpicture}
\caption{Karp-Rabin Fingerprints for second-level subpattern.}\figlab{fig:2nd:subpattern}
\end{figure*}
We also use the following application of the Prime Number Theorem from  \cite{CliffordFPSS16}:

\newcommand{\lemhash}{Given two distinct integers $a,b\in[n]$ and a random prime number $p\in\left[\frac{d}{\beta}\log^2 n,\frac{34d}{\beta}\log^2 n\right]$ where $\beta=\frac{1}{16}$, then $\PPr{a\equiv b\bmod{p}}\le\frac{\beta}{32d}$.}
\begin{lemma}(Adaptation of Lemma 4.1 \cite{CliffordFPSS16})\lemlab{lem:hash}
\lemhash
\end{lemma}
\begin{proof}
By the Prime Number Theorem  (Corollary $1$ of \cite{RosserS62}) it follows that the number of primes in $\left[\frac{d}{\beta}\log^2 n,\frac{34d}{\beta}\log^2 n\right]$ is at least
\[\frac{\frac{(34-2)d}{\beta}\log^2 n}{\log\left(\frac{34d}{\beta}\log^2 n\right)}\ge\frac{\frac{32d}{\beta}\log^2 n}{\log n}\ge\frac{32d}{\beta}\log n.\]
If $a\equiv b\bmod{p}$, then $p$ is a divisor of $|a-b|$. 
Furthermore, by assumption, $p$ is prime.  
Thus, the probability that $p$ is one of the prime divisors of $|a-b|\le n-1$ is at most $\frac{\log n}{(32d/\beta)\log n}=\frac{\beta}{32d}$, since $|a-b|$ can have at most $\log n$ prime divisors.
\end{proof}

Finally, we remark that problems in bioinformatics, such as the \emph{RNA Folding Problem}~\cite{SublinearOpen61}, use the following notion of complementary palindromes:
\begin{definition}
Let $f\,:\,\sum\rightarrow\sum$ be a pairing of symbols in the alphabet. A string $S\in\sum^n$ is a complementary palindrome if $S[x]=f(S[n+1-x])$ for all $1\le x\le n$.
\end{definition}

Our algorithms can be modified to recognize complementary palindromes with the same space usage and update time.
Indeed, we only need to modify the forward fingerprints to use $f(S[x])$ instead of $S[x]$:
\[\phi^F_{a,b}(S)=\left(\sum_{x\equiv a\bmod{b}}f(S[x])\cdot B^{\ceil{x/b}}\right)\bmod{P}.\]
\section{Overview and Techniques}
\subsection*{One-pass Multiplicative Approximation Algorithm} 
Our algorithm combines and extends ideas and techniques from the solution to the $d$-Mismatch Problem in \cite{CliffordFPSS16} 
and the solution to the Longest Palindrome Problem in \cite{BerenbrinkEMA14}. 

As the stream progresses, we keep a set of checkpoints $\mathcal{C}$, where each $c\in\mathcal{C}$ is an index for which we search $d$-near-palindromes to begin.   
We also maintain a sliding window that contains the $2d$ most recently seen symbols, as shown in \figref{fig:window}.  
The sliding window identifies any $d$-near-palindrome of length at most $2d$. 
It also guesses that the midpoint of the sliding window is the midpoint of a potential $d$-near-palindrome of length  $>2d$.
\begin{figure*}[htb]
\centering
\begin{tikzpicture}[scale=0.6]
\draw (1cm,0cm) -- (19cm,0cm);
\draw (1cm,1.2cm) -- (19cm,1.2cm);

\draw (1cm,0cm) -- (1cm,1.2cm);
\draw (1.6cm,0cm) -- (1.6cm,1.2cm);
\draw (2.2cm,0cm) -- (2.2cm,1.2cm);
\draw (2.8cm,0cm) -- (2.8cm,1.2cm);
\draw (3.4cm,0cm) -- (3.4cm,1.2cm);
\draw (4cm,0cm) -- (4cm,1.2cm);
\draw (4.6cm,0cm) -- (4.6cm,1.2cm);
\draw (5.2cm,0cm) -- (5.2cm,1.2cm);
\draw (5.8cm,0cm) -- (5.8cm,1.2cm);
\draw (6.4cm,0cm) -- (6.4cm,1.2cm);
\draw (7cm,0cm) -- (7cm,1.2cm);
\draw (7.6cm,0cm) -- (7.6cm,1.2cm);
\draw (8.2cm,0cm) -- (8.2cm,1.2cm);
\draw (8.8cm,0cm) -- (8.8cm,1.2cm);
\draw (9.4cm,0cm) -- (9.4cm,1.2cm);
\draw (10cm,0cm) -- (10cm,1.2cm);
\draw (10.6cm,0cm) -- (10.6cm,1.2cm);
\draw (11.2cm,0cm) -- (11.2cm,1.2cm);
\draw (11.8cm,0cm) -- (11.8cm,1.2cm);
\draw (12.4cm,0cm) -- (12.4cm,1.2cm);
\draw (13cm,0cm) -- (13cm,1.2cm);
\draw (13.6cm,0cm) -- (13.6cm,1.2cm);
\draw (14.2cm,0cm) -- (14.2cm,1.2cm);
\draw (14.8cm,0cm) -- (14.8cm,1.2cm);
\draw (15.4cm,0cm) -- (15.4cm,1.2cm);
\draw (16cm,0cm) -- (16cm,1.2cm);
\draw (16.6cm,0cm) -- (16.6cm,1.2cm);
\draw (17.2cm,0cm) -- (17.2cm,1.2cm);
\draw (17.8cm,0cm) -- (17.8cm,1.2cm);
\draw (18.4cm,0cm) -- (18.4cm,1.2cm);
\draw (19cm,0cm) -- (19cm,1.2cm);
\draw[dashed] (7cm,-0.6cm) -- (7cm,1.8cm);
\draw[dashed] (13cm,-0.6cm) -- (13cm,1.8cm);
\draw[dashed] (7cm,-0.6cm) -- (13cm,-0.6cm);
\draw[dashed] (7cm,1.8cm) -- (13cm,1.8cm);
\node at (10cm,2.2cm){sliding window $2d$};
\end{tikzpicture}
\caption{We maintain a sliding window of size $2d$ to identify $d$-near-palindromes of small length.}\figlab{fig:window}
\end{figure*}
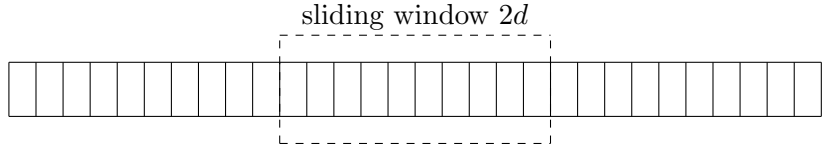
We keep an estimate $\tilde{\ell}$ of the length  $\ell_{max}$ of the longest $d$-near-palindrome seen throughout the stream, as well as its starting index $c_{start}$, and the locations of the mismatches, a set of size at most $d$. 
Upon reading symbol $S[x]$ of the stream, we call procedure $\check$ to see if $S[c_i,x]$ is a $d$-near-palindrome, for each checkpoint $c_i$ such that $x-c_i>\tilde{\ell}$, as in \figref{fig:check}. 
\begin{figure*}[htb]
\centering
\begin{tikzpicture}[scale=0.6]
\draw (1cm,0cm) -- (19cm,0cm);
\draw (1cm,1.2cm) -- (19cm,1.2cm);

\draw (1cm,0cm) -- (1cm,1.2cm);
\draw (1.6cm,0cm) -- (1.6cm,1.2cm);
\draw (2.2cm,0cm) -- (2.2cm,1.2cm);
\draw (2.8cm,0cm) -- (2.8cm,1.2cm);
\draw (3.4cm,0cm) -- (3.4cm,1.2cm);
\draw (4cm,0cm) -- (4cm,1.2cm);
\draw (4.6cm,0cm) -- (4.6cm,1.2cm);
\draw (5.2cm,0cm) -- (5.2cm,1.2cm);
\draw (5.8cm,0cm) -- (5.8cm,1.2cm);
\draw (6.4cm,0cm) -- (6.4cm,1.2cm);
\draw (7cm,0cm) -- (7cm,1.2cm);
\draw (7.6cm,0cm) -- (7.6cm,1.2cm);
\draw (8.2cm,0cm) -- (8.2cm,1.2cm);
\draw (8.8cm,0cm) -- (8.8cm,1.2cm);
\draw (9.4cm,0cm) -- (9.4cm,1.2cm);
\draw (10cm,0cm) -- (10cm,1.2cm);
\draw (10.6cm,0cm) -- (10.6cm,1.2cm);
\draw (11.2cm,0cm) -- (11.2cm,1.2cm);
\draw (11.8cm,0cm) -- (11.8cm,1.2cm);
\draw (12.4cm,0cm) -- (12.4cm,1.2cm);
\draw (13cm,0cm) -- (13cm,1.2cm);
\draw (13.6cm,0cm) -- (13.6cm,1.2cm);
\draw (14.2cm,0cm) -- (14.2cm,1.2cm);
\draw (14.8cm,0cm) -- (14.8cm,1.2cm);
\draw (15.4cm,0cm) -- (15.4cm,1.2cm);
\draw (16cm,0cm) -- (16cm,1.2cm);
\draw (16.6cm,0cm) -- (16.6cm,1.2cm);
\draw (17.2cm,0cm) -- (17.2cm,1.2cm);
\draw (17.8cm,0cm) -- (17.8cm,1.2cm);
\draw (18.4cm,0cm) -- (18.4cm,1.2cm);
\draw (19cm,0cm) -- (19cm,1.2cm);

\draw (4cm,-0.3cm) -- (4cm,-0.7cm);
\node at (4cm, -1.2cm){$c_1$};
\draw [decorate,decoration={brace,mirror,amplitude=10pt}](16cm,2.6cm) -- (4cm,2.6cm);
\draw [dashed] (4cm,1.2cm) -- (4cm, 2.6cm);
\node at (10cm, 3.6cm){$\check(c_1,x)$};
\draw [decorate,decoration={brace,mirror,amplitude=10pt}](16cm,1.2cm) -- (8.8cm,1.2cm);
\draw (8.8cm,-0.3cm) -- (8.8cm,-0.7cm);
\node at (8.8cm, -1.2cm){$c_2$};
\node at (12.7cm, 2.2cm){\tiny{$\check(c_2,x)$}};
\draw (16cm,-0.3cm) -- (16cm,-0.7cm);
\node at (16cm, -1.2cm){$x$};
\end{tikzpicture}
\caption{Upon reading $S[x]$, we check each checkpoint $c_i$ whether $S[c_i,x]$ is a $d$-near-palindrome.}\figlab{fig:check}
\end{figure*}
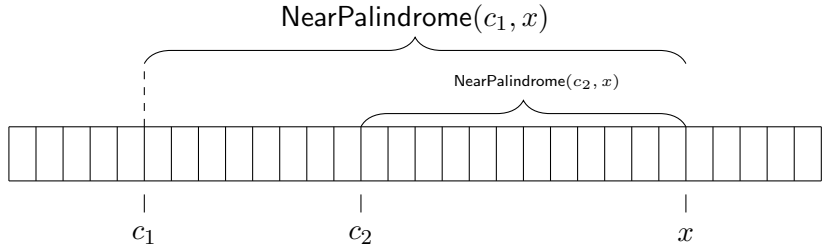
Using the framework of \cite{BerenbrinkEMA14}, we create and update checkpoints throughout the stream so that we find a $d$-near-palindrome of length at least $\frac{\ell_{max}}{1+\eps}$, as in \figref{fig:sandwich}.
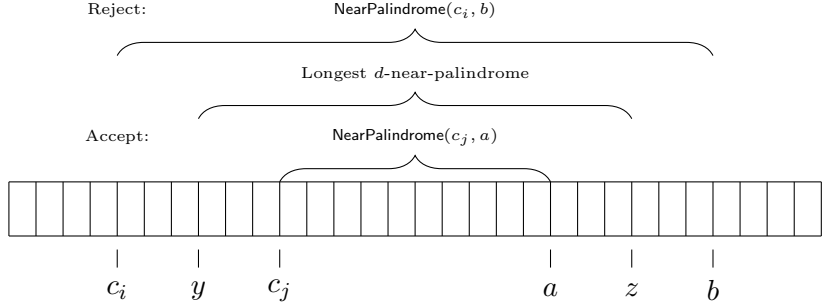
\begin{figure*}[htb]
\centering
\begin{tikzpicture}[scale=0.6]
\draw (1cm,0cm) -- (19cm,0cm);
\draw (1cm,1.2cm) -- (19cm,1.2cm);

\draw (1cm,0cm) -- (1cm,1.2cm);
\draw (1.6cm,0cm) -- (1.6cm,1.2cm);
\draw (2.2cm,0cm) -- (2.2cm,1.2cm);
\draw (2.8cm,0cm) -- (2.8cm,1.2cm);
\draw (3.4cm,0cm) -- (3.4cm,1.2cm);
\draw (4cm,0cm) -- (4cm,1.2cm);
\draw (4.6cm,0cm) -- (4.6cm,1.2cm);
\draw (5.2cm,0cm) -- (5.2cm,1.2cm);
\draw (5.8cm,0cm) -- (5.8cm,1.2cm);
\draw (6.4cm,0cm) -- (6.4cm,1.2cm);
\draw (7cm,0cm) -- (7cm,1.2cm);
\draw (7.6cm,0cm) -- (7.6cm,1.2cm);
\draw (8.2cm,0cm) -- (8.2cm,1.2cm);
\draw (8.8cm,0cm) -- (8.8cm,1.2cm);
\draw (9.4cm,0cm) -- (9.4cm,1.2cm);
\draw (10cm,0cm) -- (10cm,1.2cm);
\draw (10.6cm,0cm) -- (10.6cm,1.2cm);
\draw (11.2cm,0cm) -- (11.2cm,1.2cm);
\draw (11.8cm,0cm) -- (11.8cm,1.2cm);
\draw (12.4cm,0cm) -- (12.4cm,1.2cm);
\draw (13cm,0cm) -- (13cm,1.2cm);
\draw (13.6cm,0cm) -- (13.6cm,1.2cm);
\draw (14.2cm,0cm) -- (14.2cm,1.2cm);
\draw (14.8cm,0cm) -- (14.8cm,1.2cm);
\draw (15.4cm,0cm) -- (15.4cm,1.2cm);
\draw (16cm,0cm) -- (16cm,1.2cm);
\draw (16.6cm,0cm) -- (16.6cm,1.2cm);
\draw (17.2cm,0cm) -- (17.2cm,1.2cm);
\draw (17.8cm,0cm) -- (17.8cm,1.2cm);
\draw (18.4cm,0cm) -- (18.4cm,1.2cm);
\draw (19cm,0cm) -- (19cm,1.2cm);

\draw [decorate,decoration={brace,mirror,amplitude=10pt}](16.6cm,4cm) -- (3.4cm,4cm);
\draw (3.4cm,-0.3cm) -- (3.4cm,-0.7cm);
\node at (3.4cm, -1.2cm){$c_i$};
\draw (16.6cm,-0.3cm) -- (16.6cm,-0.7cm);
\node at (16.6cm, -1.2cm){$b$};
\node at (10cm, 5cm){\tiny{$\check(c_i,b)$}};
\draw [decorate,decoration={brace,mirror,amplitude=10pt}](14.8cm,2.6cm) -- (5.2cm,2.6cm);
\node at (10cm, 3.6cm){\tiny{Longest $d$-near-palindrome}};
\draw (5.2cm,-0.3cm) -- (5.2cm,-0.7cm);
\node at (5.2cm, -1.2cm){$y$};
\draw (14.8cm,-0.3cm) -- (14.8cm,-0.7cm);
\node at (14.8cm, -1.2cm){$z$};
\draw [decorate,decoration={brace,mirror,amplitude=10pt}](13cm,1.2cm) -- (7cm,1.2cm);
\draw (7cm,-0.3cm) -- (7cm,-0.7cm);
\node at (7cm, -1.2cm){$c_j$};
\node at (10cm, 2.2cm){\tiny{$\check(c_j,a)$}};
\draw (13cm,-0.3cm) -- (13cm,-0.7cm);
\node at (13cm, -1.2cm){$a$};
\node at (3.4cm, 2.2cm){\tiny{Accept:}};
\node at (3.4cm, 5cm){\tiny{Reject:}};
\end{tikzpicture}
\caption{The longest $d$-near-palindrome will be sandwiched within checkpoints to provide a $(1+\eps)$-approximation of $\ell_{max}$. That is, $(1+\eps)(a-c_j)\ge(z-y)$.}\figlab{fig:sandwich}
\end{figure*}

The algorithm in \cite{BerenbrinkEMA14} also maintains a list of potential midpoints associated with each checkpoint. 
Although this list can be linear in size, it satisfies nice structural results that can be used to succinctly represent the list of candidate midpoints. 
However, directly adapting these structural results to our setting would incur an extra factor of $d$ in our space complexity.
We avoid this extra factor by circumventing the list of candidate midpoints in the one-pass algorithms altogether.

We now overview the procedure $\check$ that we use repeatedly in our algorithms. 
The procedure returns whether $S[c_i,x]$ is a $d$-near-palindrome, and if so, it returns the corresponding mismatches.

The procedure $\check$ adapts the data structures outlined in \cite{CliffordFPSS16}.
Recall that in the $d$-Mismatch Problem, we are given a pattern $R$ and a text $S$ and the algorithm is required to output all indices $x$ such that 
$\HAM(R, S[x, x+|R|-1])\leq d$. 
While the pattern is fixed in the $d$-Mismatch Problem, here we essentially use variable-length patterns. 
Namely, we check whether $\HAM(S[c_i, x], S^R[c_i, x])\leq d$ for each checkpoint $c_i$ by maintaining dynamic sets of fingerprints. 

The procedure has two stages. 
In the first stage it eliminates strings $T$ with $\HAM(T, T^R)\ge 2d$, while in the second stage it eliminates strings with
 $d<\HAM(T, T^R)< 2d$.
This can be achieved by estimating the distance between $T$ and $T^R$ using fingerprints of equivalence classes modulo different primes.
 
Intuitively, picking random primes distributes the mismatches into different equivalence classes. 
For each prime $p$, the procedure estimates the number of mismatches by comparing the fingerprints of the substrings whose indices are in the same congruence class modulo $p$ with the reverse fingerprints, namely $T_{r,p}$ and $T^R_{r,p}$ for all $1\le r\le p$. Denote by $T_{r,p}$ and $T^R_{r,p}$ the \emph{first-level fingerprints}.

By the second stage we are only left with the strings with a small number of mismatches. In order to recover the mismatches, one needs to  refine each subpattern $\tilde{T}=T_{r,p}$  by picking smaller primes $p'$, and comparing the fingerprints of the strings $\tilde{T}_{r',p'}$ and $\tilde{T}^R_{r',p'}$ for all $1\le r'\le p'$. Denote by $\tilde{T}_{r',p'}$ and $\tilde{T}^R_{r',p'}$ the \emph{second-level fingerprints} (see \figref{fig:2nd:subpattern}).

In the first stage, we sample $2\log n$ primes uniformly at random from $\left[\frac{d}{\beta}\log^2 n,\frac{34d}{\beta}\log^2 n\right]$, where $\beta=1/16$. 
Each prime generates $p$ subpatterns containing positions in the same equivalence class$\pmod{p}$.
Therefore, there are $\bigO{d\,\log^3 n}$ first-level subpatterns. 
In the second stage, we take all primes in $[\log n, 3\log n]$ that together with the primes picked in the first stage generate a total of 
$\bigO{d\,\log^5 n}$ second-level subpatterns.

Finally, we assume throughout the paper that the fingerprints of any subpattern do not fail. 
Since there are at most $n^3$ subpatterns, and the probability that a particular fingerprint fails is at most $\frac{1}{n^5}$ for $P\in[n^5,n^6]$ (by Theorem 1 in \cite{BreslauerG14}), then by a union bound, the probability that no fingerprint fails is least $1-\frac{1}{n^2}$.

Our choice of parameters is more space-efficient compared to the data structure given by \cite{CliffordFPSS16}, which uses $\bigO{d^2\log^7 n}$ space, since we no longer need the starting index to slide. 

We also note that \cite{PoratL07} gives another data structure for determining the Hamming distance between two strings. 
That data structure is more space efficient than the data structure above given by \cite{CliffordFPSS16}, but seemingly does not suffice for our problem, as it does not support concatenation, which is needed for maintaining the checkpoints. 

\subsection*{One-pass Additive Approximation and Two-Pass Exact Algorithms}
To obtain the one-pass additive approximation, we modify our checkpoints, so that they appear in every $\flr{\frac{E}{2}}$ positions. 
Hence, the longest $d$-near-palindrome must have some checkpoint within $\flr{\frac{E}{2}}$ positions of it, and the algorithm will recover a $d$-near-palindrome with length at least $\ell_{max}-E$.

To obtain the two-pass {\em exact} algorithm, we set $E=\sqrt{n}$ and modify the additive error algorithm so that it returns a list $\mathcal{L}$ of candidate midpoints of $d$-near-palindromes. 
Moreover, we show a structural result in \lemref{lem:periodic}, which allows us to compress certain substrings in the first pass, so that the second pass can recover mismatches for any potential $d$-near-palindromes within these substrings.

In the second pass, we carefully keep track of the $\frac{\sqrt{n}}{2}$ characters before the starting positions of long $d$-near-palindromes identified in the first pass. 
We use the compressed information from the first pass to reconstruct the fingerprints and calculate the number of mismatches within these long $d$-near-palindromes identified in the first pass. 
However, the actual $d$-near-palindromes may extend beyond the estimate returned in the first pass. 
Thus, we compare the $\frac{\sqrt{n}}{2}$ characters after the $d$-near-palindromes identified in the first pass with the $\frac{\sqrt{n}}{2}$ characters that we track. 
This allows us to exactly identify the longest $d$-near-palindrome during the second pass.

\subsection*{Lower Bounds} 
To show lower bounds for randomized algorithms solving the $d$-near palindrome problem we use Yao's Principle \cite{Yao77}, and construct distributions for which any deterministic algorithm fails with significant probability unless given a certain amount of space. 
We first show that providing a $(1+\eps)$ approximation to the length of longest $d$-near-palindromes inheritly solves the problem of exactly identifying whether two strings have Hamming distance at most $d$. 
This problem has been useful in proving other related lower bounds \cite{ErgunGSZ17, GrigorescuSZ17} and may be of independent interest. 
We carefully construct hard distributions for this problem, using ideas from \cite{GawrychowskiMSU16}, and show via counting arguments that deterministic algorithms using a little of space will fail with significant probability on inputs from these distributions. 
\section{One-Pass Streaming Algorithm with Multiplicative Error $(1+\eps)$}
In this section, we prove \thmref{thm:ham}. 
Namely, we provide a one-pass streaming algorithm with multiplicative error $(1+\eps)$, using  $\bigO{\frac{d\log^7 n}{\eps\log(1+\eps)}}$ bits of space.
\subsection{Algorithm}
\seclab{sec:1p:alg}
As described in the overview, similar to \cite{BerenbrinkEMA14}, we maintain a sliding window of size $2d$, along with master fingerprints, and a series of checkpoints.
From the sliding window, we observe every $d$-near-palindrome with length at most $2d$, as well as every candidate midpoints.
Then, prior to seeing element $S[x]$ in the stream, we initialize the following in memory:

\begin{mdframed}
Initialization:
\begin{enumerate}
\item
Pick a prime $P$ from $[n^5,n^6]$ and an integer $B<P$ (the modulo and the base of the Karp-Rabin fingerprints, respectively).
\item
For the first-level fingerprints, create set $\mathcal{P}$ consisting of $2\log n$ primes $p_1,p_2,\ldots,p_{2\log n}$ sampled independently and uniformly at random from $\left[\frac{d}{\beta}\log^2 n,\frac{34d}{\beta}\log^2 n\right]$, where $\beta=\frac{1}{16}$.
\item
For the second-level fingerprints, let $\mathcal{Q}$ be the set of primes in $[\log n,3\log n]$.
\item
Initialize a sliding window of size $2d$.
\item
Initialize the sets of \emph{Master Fingerprints}, $\mathcal{F}^F$ and $\mathcal{F}^R$:
\begin{enumerate}
\item
Set $\phi^F_{r,p}(S)=0$, $\phi^R_{r,p}(S)=0$ for all $p\in\mathcal{P}$ and $1\le r\le p$.
\item
Set $\phi^F_{r',pq}(S)=0$, $\phi^R_{r',pq}(S)=0$ for all $p\in\mathcal{P}$, $q\in\mathcal{Q}$ and $1\le r'\le pq$.
\item
Let $\mathcal{F}^F$ be the set of all $\phi^F(S)$.
\item
Let $\mathcal{F}^R$ be the set of all $\phi^R(S)$.
\end{enumerate}
\item
Set $k_0=\frac{\log(1/\alpha)}{\log(1+\alpha)}$, where $\alpha=\sqrt{1+\eps}-1$.
\item
Initialize a list of checkpoints $\mathcal{C}=\emptyset$.
\item
Set the starting index $c_{start}$ to be $1$, the length estimate $\tilde{\ell}$ of the longest $d$-near-palindrome found so far to be $0$, and the at most $d$ mismatched indices $\mathcal{M}=\emptyset$.
\end{enumerate}
\end{mdframed}

We now formalize the steps outlined in the overview. The data structure relies on the procedure \check that we describe and analyze in detail in \secref{sec:check}.

\begin{mdframed}
Maintenance:
\begin{enumerate}
\item
Read $S[x]$. Update the sliding window to $S[x-2d, x]$.
\item
Update the \emph{Master Fingerprints} to be $\mathcal{F}^F(1,x)$ and $\mathcal{F}^R(1,x)$:
\begin{enumerate}
\item
Update the first-level fingerprints: for every $p\in\mathcal{P}$, let $r\equiv x \mod p$, and increment $\phi_{r,p}^F(S)$ by $S[x]\cdot B^{\ceil{x/p}}\bmod{P}$ and increment $\phi_{r,p}^R(S)$ by $S[x]\cdot B^{-\ceil{x/p}}\bmod{P}$.
\item
Update the second-level fingerprints: for every $p\in\mathcal{P}$ and $q\in\mathcal{Q}$, let $r'\equiv x\mod pq$, and increment $\phi_{r',pq}^F(S)$ by $S[x]\cdot B^{\ceil{x/(pq)}}\bmod{P}$ and increment $\phi^R_{r',pq}(S)$ by $S[x]\cdot B^{-\ceil{x/(pq)}}\bmod{P}$.
\end{enumerate}
\item
For all $k\ge k_0$:
\begin{enumerate}
\item
If $x$ is a multiple of $\flr{\alpha(1+\alpha)^{k-2}}$, then add the checkpoint $c=x$ to $\mathcal{C}$. Set $\level{c}=k$, 
$\fingerprint{c}={\cal F}^F(1,x)\cup {\cal F}^R(1,x)$.
\item
If there exists a checkpoint $c$ with $\level{c}=k$ and $c<x-2(1+\alpha)^k$, then delete $c$ from $\mathcal{C}$.
\end{enumerate}
\item
For every checkpoint $c\in\mathcal{C}$ such that $x-c>\tilde{\ell}$, we call  $\check$ (described in \secref{sec:check}) to see if $S[c,x]$ is a $d$-near-palindrome. 
If $S[c,x]$ is a $d$-near-palindrome, then set $c_{start}=c$, $\tilde{\ell}=x-c$ and $\mathcal{M}$ to be the indices returned by $\check$.
\item
If $x=n$, then report $c_{start}$, $\tilde{\ell}$, and $\mathcal{M}$.
\end{enumerate}
\end{mdframed}
\subsection{Procedure \check and Analysis}
\seclab{sec:check}
In this section, we describe and analyze the randomized procedure $\check$ that receives as input a string, and decides whether it is a $d$-near-palindrome or not. Moreover, if the string is a $d$-near-palindrome, $\check$ returns the locations of the mismatched indices. As mentioned,  $\check$ adapts ideas  from  \cite{CliffordFPSS16}  for solving the $k$-mismatch problem.
Our proofs of the properties of $\check$ follow almost verbatim from the statements in \cite{CliffordFPSS16}, with the only difference being that we make the magnitudes of the chosen primes as large as to withstand patterns of length $\bigO{n}$. We also  use the notations from \cite{CliffordFPSS16}, which we introduce next.

Given a string $S[x,y]$, and prime $p_j$ let $\Delta_j(x,y)$ be the number of $r\in [p_j]$ such that the subpatterns $S_{r,p_j}[x,y]$ and $S^R_{r,p_j}[x,y]$ are different. 
Note that we can compute $\Delta_j(x,y)$ from the fingerprints $\mathcal{F}^F(x,y)$ and $\mathcal{F}^R(x,y)$ as the number of indices $r$ such that $\phi^F_{r,p_j}[x,y]\neq B^{k+1}\cdot\phi^R_{r,p_j}[x,y]\bmod{P}$, where $k$ is the length of $S_{r,p_j}[x,y]$. Define $\Delta(x,y)=\max_j\Delta_j(x,y)$. 
We may assume throughout that $S[x,y]$ has even length. 
Next we summarize some useful properties of $\Delta(x,y)$.

\newcommand{\lemreject}{Let $\beta=1/16$.
\begin{enumerate}
\item \label{rej-1}  If $\HAM(S[x,y], S^R[x,y])\le d$, then $\Delta(x,y)\le d$.
\item \label{rej-2} If $\HAM(S[x,y], S^R[x,y])\geq 2d$, then $\Delta(x,y)>(1+\beta)\cdot d$ with probability at least $1-\frac{1}{n^3}$.
\end{enumerate}
}
\begin{lemma}(Adaptation of Lemma 5.1 and Lemma 5.2 \cite{CliffordFPSS16})\lemlab{lem:reject}
\lemreject
\end{lemma}
\begin{proof}
Recall that $\Delta(x,y)=\max_j\Delta_j(x,y)$, where $\Delta_j(x,y)$ is the number of indices $r$ such that the subpatterns $S_{r,p_j}[x,y]$ and $S^R_{r,p_j}[x,y]$ are not the same. Also recall that a mismatch is an index $a$ s.t. $S[x+a]\ne S[y-x-a+1]$. 
Then (\ref{rej-1}) follows from the observation that for every $p_j$, the number of  $r\in [p_j]$ for which  $S_{r,p_j}[x,y]\ne S^R_{r,p_j}[x,y]$  is at most the number of mismatches of $S[x,y]$, and so $\HAM(S[x,y], S^R[x,y]) \geq \Delta(x,y)$. 
We now show that $\Delta(x,y)\le(1+\beta)\cdot d$ w.p. $\le\frac{1}{n^3}$, thus proving (\ref{rej-2}).
Assume that $\HAM(S[x,y], S^R[x,y])\geq 2d$,  and let $\mathcal{M}$ be any set of $2d$ mismatches in $S[x,y]$. 
 
A mismatch $a$  is  {\em $\mathcal{M}$-isolated} under prime $p_j$ if there exists some $r\in[p_j]$ so that $a$  is the only mismatch from $\mathcal{M}$  in the first-level subpattern $S_{r,p_j}[x,y]$. 
Hence, the number of $\mathcal{M}$-isolated mismatches under any prime $p_j$ is a lower bound on $\Delta(x,y)$.

We will show that $\PPr{\Delta(x,y)<(1+\beta)d}\leq 1/n^3.$ 

\begin{claim}
We have $Pr_{p}[\Delta_j(x,y)<(1+\beta)d]<1/8$, over random prime $p$ chosen by the algorithm.
\end{claim}
\begin{proof}
Note that $\Delta_j(x,y)<(1+\beta)d$ if and only if at least $(1-\beta)d$ elements in $\mathcal{M}$ are not $\mathcal{M}$-isolated under $p_j$. 
  By \lemref{lem:hash}, for any $a, b\in\mathcal{M}$, the probability $a\equiv b\pmod{p_j}$ is at most $\frac{\beta}{32d}$.
  Therefore, by a union bound, for a fixed $a\in\mathcal{M}$, $a$ is not $\mathcal{M}$-isolated under $p_j$ w.p. $\frac{\beta}{32d}\cdot (2d)=\beta/16$.
  Thus, the expected number of elements in $\mathcal{M}$ that are not $\mathcal{M}$-isolated is less than $(\beta d)/8$.
  By Markov's inequality, the number of elements in $\mathcal{M}$ that are not $\mathcal{M}$-isolated exceeds $(1-\beta)d$ with probability at most $\beta/(8\cdot (1-\beta))<1/(8\cdot 15)<1/8$.
\end{proof}

From the claim and from the fact that $\Delta(x, y)=\max \Delta_j(x,y)$ it follows that after picking $2\log n$ random primes in $\left[\frac{d}{\beta}\log^2 n,\frac{34d}{\beta}\log^2 n\right]$ we have $\PPr{\Delta(x,y)<(1+\beta)d}\leq (1/8)^{2\log n}\leq n^{-3}$. 
\end{proof}

A position $i\in[x,y]$ is an \emph{isolated mismatch} under $p_j$ if there exists some $r\le p_j$ for which the subpatterns $S_{r,p_j}[x,y]$ and $S^R_{r,p_j}[x,y]$ differ only in position $i$. 
Let $\mathcal{I}_j(x,y)$  be the number of isolated mismatches in $S[x,y]$ under $p_j$, and let $\mathcal{I}(x,y)$ be the union of $\mathcal{I}_j(x,y)$, over all primes $p_j$. 
The next lemma shows that if  $\HAM(S[x,y], S^R[x,y])\leq 2d$ , then  $\mathcal{I}(x,y)$ is precisely $\HAM(S[x,y], S^R[x,y])$ with high probability over the set of primes.

\newcommand{\lemaccept}{If $\HAM(S[x,y], S^R[x,y])\le 2d$, then 
\newline\noindent $\HAM(S[x,y], S^R[x,y])=\mathcal{I}(x,y)$ with probability at least $1-\frac{1}{n^7}$.
}
\begin{lemma}(Adaptation of Lemma 4.2 \cite{CliffordFPSS16})\lemlab{lem:accept}
\lemaccept
\end{lemma}
\begin{proof}
Since $\mathcal{I}(x,y)$ is the union of $\mathcal{I}_j(x,y)$, the number of isolated mismatches in $S[x,y]$ under $p_j$, then $\HAM(S[x,y], S^R[x,y])=\mathcal{I}(x,y)$ if and only if each mismatch is isolated under $p_j$ for some $j$. 

For fixed $a,b\in\mathcal{M}'$, the probability that $a\equiv b\bmod{p_j}$ is at most $1/32d$ by \lemref{lem:hash}.
As before, since there are at most $2d$ mismatches, the probability that $a\equiv b\bmod{p_j}$ for some $b\in\mathcal{M}'$ is at most $1/16$ by a union bound. 
This is the probability that $a$ is not isolated under $p_j$.

Thus, the probability that $a$ is not isolated under any of the random $2\log n$ primes in $\mathcal{P}$ is at most $(1/16)^{2\log n}=1/n^8$. 
Thus, the probability that there is some $a\in\mathcal{M}'$ that is not isolated under any of the primes is at most  $2d/n^8\le 1/n^7$, by another union bound.

Recall that if all mismatches are isolated, then $\HAM(S[x,y], S^R[x,y])=\mathcal{I}(x,y)$, and so the probability that $\HAM(S[x,m], S^R[m+1,y])\neq\mathcal{I}(x,y)$ is at most $1/n^7$. 
\end{proof}

\begin{lemma}\lemlab{lem:computemismatch}(Adaptation of Lemma 4.3 \cite{CliffordFPSS16})
The set of mismatches can be identified using the second-level fingerprints.
\end{lemma}
\begin{proof}
Using the notion from the proof of  \lemref{lem:accept}, note that if subpattern $S_{r_j,p_j}[x,y]$ contains an isolated mismatch for prime $p_j$ and $r_j\in[p_j]$, then this mismatch is exactly the one position that does not match in the second-level subpattern. 
It remains to show that the algorithm can correctly recover the isolated mismatch through the second-level subpatterns.
Suppose, by way of contradiction, the algorithm recovers some index $s$ not equivalent to the mismatch $t$ isolated under $p_j$. 
Then it follows that both $s$ and $t$ are equivalent to $r_1\bmod{p_jq_1}$, $r_2\bmod{p_jq_2}$, $\ldots$, $r_{|\mathcal{Q}|}\bmod{p_jq_{|\mathcal{Q}|}}$. 
By Theorem 1 of \cite{RosserS62}, the product of the primes $q_i$ is at least $n$. 
Thus, by the Chinese Remainder Theorem, $s=t$, which is a contradiction. 
It follows that the algorithm correctly identifies the location of any isolated mismatches. 
\end{proof}
We are now ready to present the algorithm in full.
\begin{mdframed}
$\check(c_i,x)$: (determines if $S[c_i,x]$ is a $d$-near-palindrome)
\begin{enumerate}
\item
For each $j\in[2\log n]$, initialize $\Delta_j=0$.
\item
For each $j\in[2\log n]$ and $r\in[p_j]$:
\begin{itemize}
\item[]
If $\phi^F_{r,p_j}(S[c_i,x])\neq B^{k+1}\cdot\phi^R_{r,p_j}(S[c_i,x])\bmod{P}$, where $k$ is the length of $S_{r,p_j}[c_i,x]$, then increment $\Delta_j(c_i,x)=\Delta_j(c_i,x)+1$. 
\end{itemize}
\item
Let $\Delta(c_i,x)=\max_j\{\Delta_j(c_i,x)\}$.
\item
If $\Delta(c_i,x)>(1+\beta)\cdot d$, then we immediately reject $S[c_i,x]$. (Recall that $\beta=\frac{1}{16}$.)
\item
Initialize $\mathcal{I}=\emptyset$.
\item
For each mismatch in $S[c_i,x]$, if there exists $q\in\mathcal{Q}$ and such that $\phi^F_{r',q}(S_{r,p}[c_i,x])\neq B^{k'+1}\cdot\phi^R_{r',q}(S_{r,p}[c_i,x])\bmod{P}$, where $k'$ is the length of $S_{r'+rp,pq}[c_i,x]$, for exactly one $r\in[p], r'\in[q]$, then insert the mismatch into $\mathcal{I}(c_i,x)$. (This is the set of \emph{isolated mismatches}.)
\item
If $|\mathcal{I}(c_i,x)|>d$, then we reject $S[c_i,x]$.
\item
Else, if $|\mathcal{I}(c_i,x)|\le d$, then we accept $S[c_i,x]$ and return $\mathcal{I}(c_i,x)$.
\end{enumerate}
\end{mdframed}
\begin{theorem} \thmlab{thm:nearpalindrome}
With probability at least $1-\frac{1}{n^3}$, procedure $\check$ returns whether $S[c_i,x]$ is a $d$-near-palindrome.
\end{theorem}
\begin{proof}
If $\HAM(S[c_i,x], S^R[c_i,x])>2d$, then by \lemref{lem:reject}, $\Delta(c_i,x)>(1+\beta)\cdot 2d$ with probability at least $1-\frac{1}{n^3}$ and so $\check$ will reject $S[c_i,x]$. 
Conditioned on $\HAM(S[c_i,x], S^R[c_i,x])\le 2d$, by \lemref{lem:accept} $\mathcal{I}(c_i,x)=\HAM(S[c_i,x], S^R[c_i,x])$ with 
probability at least $1-\frac{1}{n^5}$, and so if $\HAM(S[c_i,x], S^R[c_i,x])>d$ the algorithm safely rejects, and otherwise it accepts. Finally, by  \lemref{lem:computemismatch} the entire set of mismatches $\mathcal{I}(c_i,x)$ can be computed from the second-level subpattern fingerprints. 
\end{proof}

\subsection{Correctness and Space Complexity}
\label{sec:checklists}
In this section, we finish the proof of \thmref{thm:ham} by claiming correctness and analyzing the space used by the one-pass streaming algorithm described in \secref{sec:1p:alg}. 
Since we used the spacing of the checkpoints as in \cite{BerenbrinkEMA14}, we have the following properties.
\newcommand{\obscheckpoints}{At reading $S[x]$, for all $k\ge k_0=\ceil{\frac{\log\left(\frac{(1+\alpha)^2}{\alpha}\right)}{\log(1+\alpha)}}$, let $C_{x,k}=\{c\in\mathcal{C}\,|\,\level{c}=k\}$.
\begin{enumerate}
\item
$C_{x,k}\subseteq[x-2(1+\alpha)^k,x]$.
\item
The distance between two consecutive checkpoints of $C_{x,k}$ is $\flr{\alpha(1+\alpha)^{k-2}}$.
\item
$|C_{x,k}|=\ceil{\frac{2(1+\alpha)^k}{\flr{\alpha(1+\alpha)^{k-2}}}}$.
\item 
At any point in the algorithm, the number of checkpoints is $\bigO{\frac{\log n}{\eps\log(1+\eps)}}$.
\end{enumerate}
}
\begin{observation}(\cite{BerenbrinkEMA14}, Observation 16, Lemma 17)\obslab{obs:checkpoints}
\obscheckpoints
\end{observation}
\begin{corollary}
The total space used by the algorithm is $\bigO{\frac{d\log^7 n}{\eps\log(1+\eps)}}$ bits. The update time per arriving symbol is also $\bigO{\frac{d\log^6 n}{\eps\log(1+\eps)}}$.
\end{corollary}
\begin{proof}
The first-level and second-level Karp-Rabin fingerprints consist of integers modulo $P$ for each of the $\bigO{d\log^5 n}$ subpatterns. 
Since $P\in[n^5,n^6]$, then $\bigO{d\log^6 n}$ bits of space are necessary for each fingerprint. 
Furthermore, by \obsref{obs:checkpoints}, there are $\frac{\log n}{\eps\log(1+\eps)}$ checkpoints, so the total space used is $\bigO{d\log^7 n}$ bits.
For each arriving symbol $S[x]$, the algorithm checks possibly the fingerprints of each checkpoint whether the substring is a $d$-near-palindrome. 
There are $\bigO{\frac{\log n}{\eps\log(1+\eps)}}$ checkpoints, each with fingerprints of size $\bigO{d\log^5 n}$.
Each subpattern of a fingerprint may be compared in constant time, so the overall update time is $\bigO{\frac{d\log^6 n}{\eps\log(1+\eps)}}$.
\end{proof}
We now show correctness and analyze the space complexity of the one-pass streaming algorithm described in \secref{sec:1p:alg}. 
\begin{proofof}{\thmref{thm:ham}}
Let $\ell_{max}$ be the length of the longest $d$-near-palindrome, $S[x,x+\ell_{max}-1]$, with midpoint $m$. 
Let $k$ be the largest integer so that $2(1+\alpha)^{k-1}<\ell_{max}$, where $\alpha=\sqrt{1+\eps}-1$. 
Let $y=m+(1+\alpha)^{k-1}$ so that $x<y<x+\ell_{max}-1$. 
By \obsref{obs:checkpoints}, there exists a checkpoint in the interval $[y-2(1+\alpha)^{k-1},y]$. 
Furthermore, \obsref{obs:checkpoints} implies consecutive checkpoints of level ${k-1}$ are separated by distance $\flr{\alpha(1+\alpha)^{k-2}}$. 
Thus, there exists a checkpoint $c$ in the interval $\left[y-2(1+\alpha)^{k-1},y-2(1+\alpha)^{k-1}+\alpha(1+\alpha)^{k-3}\right]$. 
If procedure $\check$ succeeds for this checkpoint on position $m+(m-c)$, then the output $\tilde{\ell}$ of the algorithm is at least 
\[2(m-c)\ge 2m-2y+4(1+\alpha)^{k-1}-2\alpha(1+\alpha)^{k-3}=2(1+\alpha)^{k-1}-2\alpha(1+\alpha)^{k-3}.\]
Comparing this output with $\ell_{max}$,
\[\frac{\ell_{max}}{\tilde{\ell}}\le\frac{2(1+\alpha)^k}{2(1+\alpha)^{k-1}-2\alpha(1+\alpha)^{k-3}}=\frac{(1+\alpha)^3}{(1+\alpha)^2-\alpha}\le(1+\alpha)^2=1+\eps.\]
Thus, if procedure $\check$ succeeds for all substrings then $\tilde{\ell}\le\ell_{max}\le(1+\eps)\tilde{\ell}$. 
Taking \thmref{thm:nearpalindrome} and a simple union bound over all $\bigO{n^2}$ possible substrings, procedure $\check$ succeeds for all substrings with probability at least $1/n$, and the result follows. 
\end{proofof}
\section{One-Pass Streaming Algorithm with Additive Error $E$}
\seclab{sec:additive}
In this section, we prove \thmref{thm:add}, showing a one-pass streaming algorithm which uses $\bigO{\frac{dn\log^6 n}{E}}$ bits of space.
The initialization of the algorithm is the same as that in \secref{sec:1p:alg} for the one-pass streaming algorithm with multiplicative error $(1+\eps)$.

\begin{mdframed}
Maintenance:
\begin{enumerate}
\item
Read $S[x]$. Update the sliding window to $S[x-2d, x]$.
\item
Update the \emph{Master Fingerprints} to be $\mathcal{F}^F(1,x)$ and $\mathcal{F}^R(1,x)$:
\begin{enumerate}
\item
For the first-level fingerprints: for every $p\in\mathcal{P}$, let $r\equiv x \mod p$, and increment $\phi_{r,p}^F(S)$ by $S[x]\cdot B^{\ceil{x/p}}\bmod{P}$ and increment $\phi_{r,p}^R(S)$ by $S[x]\cdot B^{-\ceil{x/p}}\bmod{P}$.
\item
For the second-level fingerprints: for every $p\in\mathcal{P}$ and $q\in\mathcal{Q}$, let $r'\equiv x\mod pq$, and increment $\phi_{r',pq}^F(S)$ by $S[x]\cdot B^{\ceil{x/(pq)}}\bmod{P}$ and increment $\phi^R_{r',pq}(S)$ by 
\newline $S[x]\cdot B^{-\ceil{x/(pq)}}\bmod{P}$.
\end{enumerate}
\item
If $x$ is a multiple of $\flr{\frac{E}{2}}$, then add the checkpoint $c=x$ to $\mathcal{C}$. Set $\fingerprint{c}={\cal F}^F(1,x)\cup {\cal F}^R(1,x)$.
\item
For every checkpoint $c\in\mathcal{C}$ such that $x-c>\tilde{\ell}$, we call procedure $\check$ to see if $S[c,x]$ is a near-palindrome. 
If $S[c,x]$ is a near-palindrome, then set $c_{start}=c$, $\tilde{\ell}=x-c$ and $\mathcal{M}$ to be the indices returned by $\check$.
\item
If $x=n$, then report $c_{start}$, $\tilde{\ell}$, and $\mathcal{M}$.
\end{enumerate}
\end{mdframed}

\begin{corollary}
The algorithm uses $\bigO{\frac{dn\log^6 n}{E}}$ bits of space and $\bigO{\frac{dn\log^5 n}{E}}$ time per arriving symbol.
\end{corollary}
\begin{proof}
Each of the Karp-Rabin fingerprints consist of $\bigO{d\log^5 n}$ integers modulo $P$. 
Since $P\in[n^5,n^6]$, then $\bigO{d\log^6 n}$ bits of space are necessary for each fingerprint. 
Each checkpoint is spaced $\flr{\frac{E}{2}}$ positions apart, so there are at most $\frac{2n}{E}+1$ checkpoints, and the total space required is $\bigO{\frac{dn\log^6 n}{E}}$ bits.
For each arriving symbol $S[x]$, the algorithm checks each checkpoint and possibly the fingerprints of each checkpoint to check whether the substring is a near-palindrome. 
Since there are $\bigO{\frac{n}{E}}$ checkpoints, each containing fingerprints of size $\bigO{d\log^5 n}$, and each subpattern of a fingerprint may be compared in constant time, then the overall update time per arriving symbol is $\bigO{\frac{dn\log^5 n}{E}}$.
\end{proof}

The correctness of the algorithm follows immediately from the spacing of the checkpoints, and the correctness of procedure $\check$.
\begin{proofof}{\thmref{thm:add}}
For each $x,y$, procedure $\check$ returns, with probability at least $1-\frac{1}{n^3}$, whether $S[c_i,x]$ is a $d$-near-palindrome. 
Thus by a simple union bound over all possible substrings of the stream, $\check$ succeeds with probability at least $1-\frac{1}{n}$. 
Because the checkpoints are separated by distance $\flr{\frac{E}{2}}$, the longest $d$-near-palindrome can begin at most $\flr{\frac{E}{2}}-1$ characters before a checkpoint. 
Hence, the algorithm outputs some $\tilde{\ell}$ such that $\tilde{\ell}\ge\ell_{max}-E$.
\end{proofof}
\section{Two-Pass  {\em Exact} Streaming Algorithm}\seclab{sec:twopass}
In this section, we prove \thmref{thm:exact}.
Namely, we present a two-pass streaming algorithm which returns the longest $d$-near-palindrome with space $\bigO{d^2\sqrt{n}\log^6 n}$. 

Recall that we assume the lengths of $d$-near-palindromes are even.
Thus, for any substring $S[x,y]$ of even length, we define its {\em midpoint} $m=\flr{\frac{x+y}{2}}$.
Upon reading $x$, we say that $x-\sqrt{n}$ is a candidate midpoint if the sliding window $S[x-2\sqrt{n},x]$ is a $d$-near-palindrome.

First, we modify the one-pass streaming algorithm with additive error in \secref{sec:additive} so that it returns a list $\mathcal{L}$ of candidate midpoints of $d$-near-palindromes with length at least $\ell-\frac{\sqrt{n}}{2}$, where $\ell$ is an estimate of the maximum length output by the algorithm.  
However, we show in \lemref{lem:periodic} that the string has a periodic structure which allows us to keep only $\bigO{d}$ fingerprints in order to recover the fingerprint for any substring between two midpoints. 

In the second pass, we explicitly keep the $\frac{\sqrt{n}}{2}$ characters before the starting positions and candidate midpoints of ``long'' $d$-near-palindromes identified in the first pass. 
We use a procedure $\recover$ to exactly identify the number and locations of mismatches within the $d$-near-palindromes identified in the first pass. 
We then use the $\frac{\sqrt{n}}{2}$ characters to extend the near-palindromes until the number of mismatches exceed $d+1$.

For an example, see \figref{fig:recover}.
\begin{figure*}[htb]
\centering
\begin{tikzpicture}[scale=0.6]
\draw (1cm,0cm) -- (19cm,0cm);
\draw (1cm,1.2cm) -- (19cm,1.2cm);

\draw (1cm,0cm) -- (1cm,1.2cm);
\draw (1.6cm,0cm) -- (1.6cm,1.2cm);
\draw (2.2cm,0cm) -- (2.2cm,1.2cm);
\draw (2.8cm,0cm) -- (2.8cm,1.2cm);
\draw (3.4cm,0cm) -- (3.4cm,1.2cm);
\draw (4cm,0cm) -- (4cm,1.2cm);
\draw (4.6cm,0cm) -- (4.6cm,1.2cm);
\draw (5.2cm,0cm) -- (5.2cm,1.2cm);
\draw (5.8cm,0cm) -- (5.8cm,1.2cm);
\draw (6.4cm,0cm) -- (6.4cm,1.2cm);
\draw (7cm,0cm) -- (7cm,1.2cm);
\draw (7.6cm,0cm) -- (7.6cm,1.2cm);
\draw (8.2cm,0cm) -- (8.2cm,1.2cm);
\draw (8.8cm,0cm) -- (8.8cm,1.2cm);
\draw (9.4cm,0cm) -- (9.4cm,1.2cm);
\draw (10cm,0cm) -- (10cm,1.2cm);
\draw (10.6cm,0cm) -- (10.6cm,1.2cm);
\draw (11.2cm,0cm) -- (11.2cm,1.2cm);
\draw (11.8cm,0cm) -- (11.8cm,1.2cm);
\draw (12.4cm,0cm) -- (12.4cm,1.2cm);
\draw (13cm,0cm) -- (13cm,1.2cm);
\draw (13.6cm,0cm) -- (13.6cm,1.2cm);
\draw (14.2cm,0cm) -- (14.2cm,1.2cm);
\draw (14.8cm,0cm) -- (14.8cm,1.2cm);
\draw (15.4cm,0cm) -- (15.4cm,1.2cm);
\draw (16cm,0cm) -- (16cm,1.2cm);
\draw (16.6cm,0cm) -- (16.6cm,1.2cm);
\draw (17.2cm,0cm) -- (17.2cm,1.2cm);
\draw (17.8cm,0cm) -- (17.8cm,1.2cm);
\draw (18.4cm,0cm) -- (18.4cm,1.2cm);
\draw (19cm,0cm) -- (19cm,1.2cm);
\filldraw[thick, shading=radial,inner color=white, outer color=gray!, opacity=1] (2.8cm,0cm) rectangle+(0.6cm,1.2cm);
\filldraw[thick, shading=radial,inner color=white, outer color=gray!, opacity=1] (3.4cm,0cm) rectangle+(0.6cm,1.2cm);
\filldraw[thick, shading=radial,inner color=white, outer color=gray!, opacity=1] (4cm,0cm) rectangle+(0.6cm,1.2cm);
\filldraw[thick, shading=radial,inner color=white, outer color=gray!, opacity=1] (4.6cm,0cm) rectangle+(0.6cm,1.2cm);
\filldraw[thick, shading=radial,inner color=white, outer color=gray!, opacity=1] (5.2cm,0cm) rectangle+(0.6cm,1.2cm);
\filldraw[thick, shading=radial,inner color=white, outer color=gray!, opacity=1] (5.8cm,0cm) rectangle+(0.6cm,1.2cm);
\filldraw[thick, shading=radial,inner color=white, outer color=gray!, opacity=1] (6.4cm,0cm) rectangle+(0.6cm,1.2cm);
\filldraw[thick, shading=radial,inner color=white, outer color=gray!, opacity=1] (13cm,0cm) rectangle+(0.6cm,1.2cm);
\filldraw[thick, shading=radial,inner color=white, outer color=gray!, opacity=1] (13.6cm,0cm) rectangle+(0.6cm,1.2cm);
\filldraw[thick, shading=radial,inner color=white, outer color=gray!, opacity=1] (14.2cm,0cm) rectangle+(0.6cm,1.2cm);
\filldraw[thick, shading=radial,inner color=white, outer color=gray!, opacity=1] (14.8cm,0cm) rectangle+(0.6cm,1.2cm);
\filldraw[thick, shading=radial,inner color=white, outer color=gray!, opacity=1] (15.4cm,0cm) rectangle+(0.6cm,1.2cm);
\filldraw[thick, shading=radial,inner color=white, outer color=gray!, opacity=1] (16cm,0cm) rectangle+(0.6cm,1.2cm);
\filldraw[thick, shading=radial,inner color=white, outer color=gray!, opacity=1] (16.6cm,0cm) rectangle+(0.6cm,1.2cm);
\draw [decorate,decoration={brace,mirror,amplitude=10pt}](13cm,1.2cm) -- (7cm,1.2cm);
\draw (13cm,-0.3cm) -- (13cm,-0.7cm);
\node at (13cm, -1.2cm){$c_i+longest(c_i)$};
\draw (7cm,-0.3cm) -- (7cm,-0.7cm);
\node at (7cm, -1.2cm){$c_i$};
\node at (10cm, 2.2cm){All mismatches kept from first-pass};
\draw [decorate,decoration={brace,mirror,amplitude=10pt}](16cm,3.6cm) -- (4cm,3.6cm);
\node at (10cm, 4.6cm){Longest $d$-near-palindrome};
\draw [dashed] (16cm,1.2cm) -- (16cm, 3.6cm);
\draw [dashed] (4cm,1.2cm) -- (4cm, 3.6cm);
\draw [decorate,decoration={brace,mirror,amplitude=10pt}](2.8cm,-1.6cm) -- (7cm,-1.6cm);
\draw (2.8cm,-0.3cm) -- (2.8cm,-0.7cm);
\node at (2.8cm, -1.2cm){$c_i-\frac{\sqrt{n}}{2}$};
\node at (4.9cm, -2.5cm){Characters kept in $\mathcal{A}$};
\draw [decorate,decoration={brace,mirror,amplitude=10pt}](13cm,-1.6cm) -- (17.2cm,-1.6cm);
\node at (15.1cm, -2.5cm){Compare these characters with those kept in $\mathcal{A}$};
\end{tikzpicture}
\caption{The second pass allows us to find the longest $d$-near-palindrome by explicitly comparing characters.}\figlab{fig:recover}
\end{figure*}
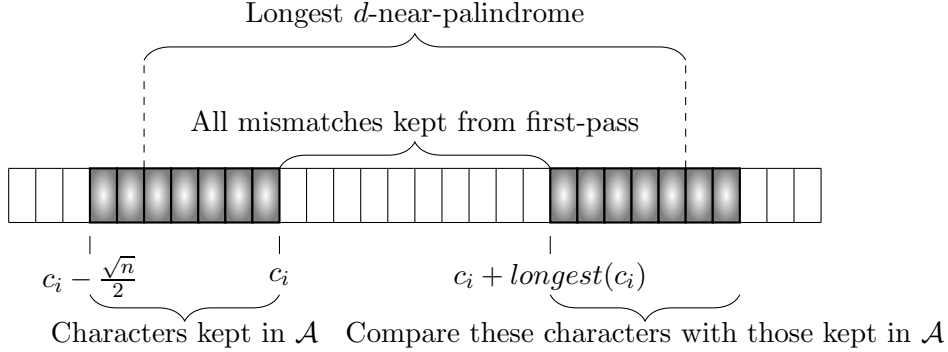
We first describe a structural property of a series of overlapping $d$-near-palindromes, showing that they are ``almost'' periodic. 
\begin{definition}
A string $S$ is said to have period $\pi$ if $S[j]=S[j+\pi]$ for all $j=1,\ldots,|S|-\pi$.
\end{definition}
The following structural result is a generalization of a structural result about palindromes from \cite{BerenbrinkEMA14} and demonstrates two properties. 
The first property shows that the midpoints of long near-palindromes are equally spaced, and thus the entire set can be represented succinctly after the first pass, even if it is linear in size. 
The second property shows a repetitive nature of the string that allows the fingerprint reconstruction of many substrings just by storing a small number of fingerprints. 
\newcommand{\lemperiodic}{Let $m_1<m_2<\ldots<m_h$ be indices in $S$ that are consecutive midpoints of $d$-near-palindromes of length $\ell^*$, for some integer $\ell^*>0$. If $m_h-m_1\le\ell^*$, then
\begin{enumerate}
\item
\label{condition1}
$m_1,m_2,\ldots,m_h$ are equally spaced in $S$, so that $|m_2-m_1|=|m_i-m_{i+1}|$ for all $i\in[h-1]$.
\item
\label{condition2}
For each $1\le i\le h$, there exists string $E_i$ with at most $d$ nonzero entries such that $E_i+S[m_1+1,m_i]$ is a prefix of $ww^Rww^R\ldots$ of length at least $\ell^*$, for some string $w$ of length $|w|=m_2-m_1$.
\end{enumerate}
}
\begin{lemma} \lemlab{lem:periodic}
\lemperiodic
\end{lemma}
\begin{proof}
Note that $m_2$ is a midpoint of a $d$-near-palindrome of length at least $\ell^*$, so there exists a string $E_2$ with at most $d$ nonzero entries such that $E_2+S[m_1+1,2m_2-m_1]$ is a palindrome of length at least $\ell^*$.

Inductively, we assume that \ref{condition1} and \ref{condition2} hold up to $m_{j-1}$. 
First, we argue that $|m_j-m_1|$ is a multiple of $|m_2-m_1|=|w|$. 
Suppose, by way of contradiction, that $m_j=m_1+|w|\cdot q+r$ for some integers $q\ge 0$ and $0<r<|w|$.
Since $m_h-m_1\le\ell^*$, then $[m_1+1,m_{j-1}+\ell^*]$ contains $m_j$. 
From our inductive hypothesis, $m_j-r$ is an index where either $w$ or $w^R$ begins. 
This implies that the prefix of $ww^R$ or $w^Rw$ of size $2r$ is a palindrome. 
By assumption, there exists $E_{j-1}$ with at most $d$ nonzero entries such that $E_{j-1}+S[m_1+1,m_h]$ is a prefix of $ww^Rww^R\ldots$ of length at least $\ell^*$. 
Thus, the interval $[m_1+1,m_1+r]$ contains a midpoint of a $d$-near-palindrome with length at least $\ell^*$. 
However, there is no such midpoint in the interval $[m_1+1,m_2-1]$, an interval with length greater than $r$, which is a contradiction. 

Thus, $m_j=m_{j-1}+|w|\cdot q$. 
Since $m_j$ is a midpoint of a $d$-near-palindrome, then \ref{condition2} follows. 
But then $m_{j-1}+|w|$ is the midpoint of a $d$-near-palindrome of length at least $\ell^*$. 
Specifically, $S[m_{j-1}+|w|-\ell^*+1,m{j-1}+|w|+\ell^*]$ is the desired $d$-near-palindrome. 
Hence, $m_j=m_{j-1}+|w|$, satisfying \ref{condition1}, and the induction is complete.
\end{proof}
In the first pass, we specify that the algorithm has sliding window size $2\sqrt{n}$. 
Thus, if the longest $d$-near-palindrome has length less than $2\sqrt{n}$, the algorithm can identify it.
Otherwise, if the longest $d$-near-palindrome has length at least $2\sqrt{n}$, then the algorithm finds at most $\frac{\sqrt{n}}{2}$ non-overlapping $d$-near-palindromes of length at least $\ell-\eps\sqrt{n}$. 
Hence, $\bigO{d^2\sqrt{n}\log^6 n}$ is enough space to store the fingerprints for the substrings between any two candidate midpoints, as well as between checkpoints $s_i\in\mathcal{L}$ and midpoints. 
The first pass of the algorithm appears below, omitting the details for when the longest $d$-near-palindrome has length at most $2\sqrt{n}$ and is therefore recognized by the sliding window.

\begin{mdframed}
First pass:
\begin{enumerate}
\item
Read $S[x]$. Set $m=x-\sqrt{n}$. Update the sliding window to $S[x-2\sqrt{n}, x]$.
\item
Update the \emph{Master Fingerprints} to be $\mathcal{F}^F(1,x)$ and $\mathcal{F}^R(1,x)$:
\begin{enumerate}
\item
For the first-level fingerprints: for every $p\in\mathcal{P}$, let $r\equiv x \mod p$, and increment $\phi_{r,p}^F(S)$ by $S[x]\cdot B^{\ceil{x/p}}\bmod{P}$ and increment $\phi_{r,p}^R(S)$ by $S[x]\cdot B^{-\ceil{x/p}}\bmod{P}$.
\item
For the second-level fingerprints: for every $p\in\mathcal{P}$ and $q\in\mathcal{Q}$, let $r'\equiv x\mod pq$, and increment $\phi_{r',pq}^F(S)$ by $S[x]\cdot B^{\ceil{x/(pq)}}\bmod{P}$ and increment $\phi^R_{r',pq}(S)$ by 
\newline $S[x]\cdot B^{-\ceil{x/(pq)}}\bmod{P}$.
\end{enumerate}
\item
If $x$ is a multiple of $\flr{\frac{\eps\sqrt{n}}{2}}$, then add the checkpoint $c=x$ to $\mathcal{C}$. Set $\fingerprint{c}=\mathcal{F}^F(1,x)\cup \mathcal{F}^R(1,x)$, $longest(c)=0$.
\item
For every checkpoint $c\in\mathcal{C}$ such that $x-c\ge\tilde{\ell}-\frac{\sqrt{n}}{2}$, we call procedure $\check$ to see if $S[c,x]$ is a near-palindrome. 
If $S[c,x]$ is a near-palindrome, then set $longest(c)=x-c$. If $x-c>\tilde{\ell}$, set $\tilde{\ell}=x-c$.
\item
If $[x-2\sqrt{n},x]$ is a $d$-near-palindrome:
\begin{enumerate}
\item
Add $m$ to $L_{c'}$, the list of candidate midpoints for the most recent checkpoint $c'$.
\item
If $|L_{c'}|=0$, store the first-level and second-level fingerprints of $S[c'+1,x]$.
\item
Else, let $m_i$ be the largest index in $L_{c'}$.
\begin{enumerate}
\item
If the first-level and second-level fingerprints of $S[m_i,m]$ match those of some other entry $S[m_j,m_{j+1}]$ stored in $L_{c'}$ and the set of indices for $m_j$ is less than $d$, add $m$ to the set of indices for $m_j$.
\item
Else, if the first-level and second-level fingerprints of $S[m_i,m]$ do not match those of any other entry stored in $L_{c'}$, then add the first-level and second-level fingerprints of $S[m_i,m]$ into $L_{c'}$, along with the index $m$.
\end{enumerate}
\end{enumerate}
\item
If $x=n$, then remove all $c\in\mathcal{C}$ such that $longest(c)<\tilde{\ell}-\sqrt{n}$.
Report $\tilde{\ell}$, $\mathcal{C}$, and $\{L_c\}$.
\end{enumerate}
\end{mdframed}
Before we proceed to the second pass, we describe procedure $\recover(m_i,m_j,L_c)$ which either outputs that $S[m_i,m_j]$ is not a $d$-near-palindrome, or returns the number of mismatches, as well as their indices. 
The procedure crucially relies on structural result from \lemref{lem:periodic} to reconstruct the fingerprints of $S[m_i,m_j]$ from fingerprints stored by the first pass. 
From the reconstructed fingerprints, the subroutine can then determine whether $S[m_i,m_j]$ is a $d$-near-palindrome, and identify the location of the mismatches, if necessary. 
The details of procedure $\recover$ in full is below:

\begin{mdframed}
$\recover(m_i,m_j,L_c)$: (determines whether $S[m_i,m_j]$ is a $d$-near-palindrome and outputs the indices and hence, number, of mismatches if it is)
\begin{enumerate}
\item
Construct the first-level and second-level fingerprints of $S[m_i,m_j]$:
\begin{enumerate}
\item
$\phi^F_{a,b}(S[m_i,m_j])=\sum_{k=i}^{j-1}B^{t}\cdot\phi^F_{a,b}(S[m_k,m_{k+1}])\bmod{P}$, where $t$ is the length of the subpattern $S_{a,b}[m_k,m_{k+1}]$.
\item
$\phi^R_{a,b}(S[m_i,m_j])=\sum_{k=i}^{j-1}B^{-t}\cdot\phi^R_{a,b}(S[m_k,m_{k+1}])\bmod{P}$, where $t$ is the length of the subpattern $S^R_{a,b}[m_k,m_{k+1}]$.
\end{enumerate}
\item
Call procedure $\check(S[m_i,m_j])$ to see whether $S[m_i,m_j]$ is a $d$-near-palindrome:
\begin{enumerate}
\item
If $S[m_i,m_j]$ is not a $d$-near-palindrome, reject $S[m_i,m_j]$.
\item
Else, accept $S[m_i,m_j]$. Output $\mathcal{I}$, the set of mismatches output by $\check$.
\end{enumerate}
\end{enumerate}
\end{mdframed}
Before the second pass, we first prune the list of checkpoints $\mathcal{C}$ to greedily include only those who are the starting indices for $d$-near-palindromes of length at least $\tilde{\ell}-\frac{\sqrt{n}}{2}$ and do not overlap with other $d$-near-palindromes already included in the list.
In the second pass, the algorithm keeps track of the $\frac{\sqrt{n}}{2}$ characters before $c$, for each starting index $c\in\mathcal{C}$. 
We call procedure $\recover$ to fully recover the mismatches in a region following $c$.
After reading the last symbol in the region, we compare each subsequent symbol with the corresponding symbol before $c$, counting the total number of mismatches. 
When the total number of mismatches reaches $d+1$ after seeing character $S[c+k+j+1]$, where $k$ is the size of the region, then the previous symbol is the end of the near-palindrome. 
Hence, the near-palindrome is $S[c-j,c+k+j]$, and if $k+2j>\tilde{l}$, then we update the information for $\tilde{\ell}$ accordingly. 
For an example, see \figref{fig:recoverb}.
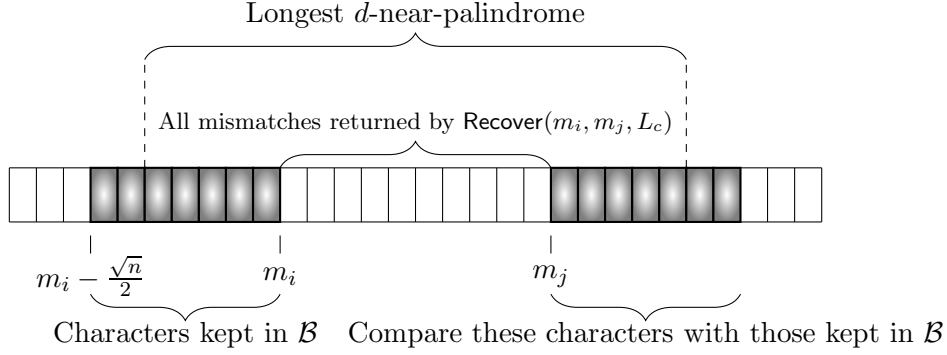
\begin{figure*}[htb]
\centering
\begin{tikzpicture}[scale=0.6]
\draw (1cm,0cm) -- (19cm,0cm);
\draw (1cm,1.2cm) -- (19cm,1.2cm);

\draw (1cm,0cm) -- (1cm,1.2cm);
\draw (1.6cm,0cm) -- (1.6cm,1.2cm);
\draw (2.2cm,0cm) -- (2.2cm,1.2cm);
\draw (2.8cm,0cm) -- (2.8cm,1.2cm);
\draw (3.4cm,0cm) -- (3.4cm,1.2cm);
\draw (4cm,0cm) -- (4cm,1.2cm);
\draw (4.6cm,0cm) -- (4.6cm,1.2cm);
\draw (5.2cm,0cm) -- (5.2cm,1.2cm);
\draw (5.8cm,0cm) -- (5.8cm,1.2cm);
\draw (6.4cm,0cm) -- (6.4cm,1.2cm);
\draw (7cm,0cm) -- (7cm,1.2cm);
\draw (7.6cm,0cm) -- (7.6cm,1.2cm);
\draw (8.2cm,0cm) -- (8.2cm,1.2cm);
\draw (8.8cm,0cm) -- (8.8cm,1.2cm);
\draw (9.4cm,0cm) -- (9.4cm,1.2cm);
\draw (10cm,0cm) -- (10cm,1.2cm);
\draw (10.6cm,0cm) -- (10.6cm,1.2cm);
\draw (11.2cm,0cm) -- (11.2cm,1.2cm);
\draw (11.8cm,0cm) -- (11.8cm,1.2cm);
\draw (12.4cm,0cm) -- (12.4cm,1.2cm);
\draw (13cm,0cm) -- (13cm,1.2cm);
\draw (13.6cm,0cm) -- (13.6cm,1.2cm);
\draw (14.2cm,0cm) -- (14.2cm,1.2cm);
\draw (14.8cm,0cm) -- (14.8cm,1.2cm);
\draw (15.4cm,0cm) -- (15.4cm,1.2cm);
\draw (16cm,0cm) -- (16cm,1.2cm);
\draw (16.6cm,0cm) -- (16.6cm,1.2cm);
\draw (17.2cm,0cm) -- (17.2cm,1.2cm);
\draw (17.8cm,0cm) -- (17.8cm,1.2cm);
\draw (18.4cm,0cm) -- (18.4cm,1.2cm);
\draw (19cm,0cm) -- (19cm,1.2cm);
\filldraw[thick, shading=radial,inner color=white, outer color=gray!, opacity=1] (2.8cm,0cm) rectangle+(0.6cm,1.2cm);
\filldraw[thick, shading=radial,inner color=white, outer color=gray!, opacity=1] (3.4cm,0cm) rectangle+(0.6cm,1.2cm);
\filldraw[thick, shading=radial,inner color=white, outer color=gray!, opacity=1] (4cm,0cm) rectangle+(0.6cm,1.2cm);
\filldraw[thick, shading=radial,inner color=white, outer color=gray!, opacity=1] (4.6cm,0cm) rectangle+(0.6cm,1.2cm);
\filldraw[thick, shading=radial,inner color=white, outer color=gray!, opacity=1] (5.2cm,0cm) rectangle+(0.6cm,1.2cm);
\filldraw[thick, shading=radial,inner color=white, outer color=gray!, opacity=1] (5.8cm,0cm) rectangle+(0.6cm,1.2cm);
\filldraw[thick, shading=radial,inner color=white, outer color=gray!, opacity=1] (6.4cm,0cm) rectangle+(0.6cm,1.2cm);
\filldraw[thick, shading=radial,inner color=white, outer color=gray!, opacity=1] (13cm,0cm) rectangle+(0.6cm,1.2cm);
\filldraw[thick, shading=radial,inner color=white, outer color=gray!, opacity=1] (13.6cm,0cm) rectangle+(0.6cm,1.2cm);
\filldraw[thick, shading=radial,inner color=white, outer color=gray!, opacity=1] (14.2cm,0cm) rectangle+(0.6cm,1.2cm);
\filldraw[thick, shading=radial,inner color=white, outer color=gray!, opacity=1] (14.8cm,0cm) rectangle+(0.6cm,1.2cm);
\filldraw[thick, shading=radial,inner color=white, outer color=gray!, opacity=1] (15.4cm,0cm) rectangle+(0.6cm,1.2cm);
\filldraw[thick, shading=radial,inner color=white, outer color=gray!, opacity=1] (16cm,0cm) rectangle+(0.6cm,1.2cm);
\filldraw[thick, shading=radial,inner color=white, outer color=gray!, opacity=1] (16.6cm,0cm) rectangle+(0.6cm,1.2cm);

\draw [decorate,decoration={brace,mirror,amplitude=10pt}](13cm,1.2cm) -- (7cm,1.2cm);
\draw (13cm,-0.3cm) -- (13cm,-0.7cm);
\node at (13cm, -1.2cm){$m_j$};
\draw (7cm,-0.3cm) -- (7cm,-0.7cm);
\node at (7cm, -1.2cm){$m_i$};
\node at (10cm, 2.2cm){\footnotesize{All mismatches returned by $\recover(m_i,m_j,L_c)$}};
\draw [decorate,decoration={brace,mirror,amplitude=10pt}](16cm,3.6cm) -- (4cm,3.6cm);
\node at (10cm, 4.6cm){Longest $d$-near-palindrome};
\draw [dashed] (16cm,1.2cm) -- (16cm, 3.6cm);
\draw [dashed] (4cm,1.2cm) -- (4cm, 3.6cm);
\draw [decorate,decoration={brace,mirror,amplitude=10pt}](2.8cm,-1.6cm) -- (7cm,-1.6cm);
\draw (2.8cm,-0.3cm) -- (2.8cm,-0.7cm);
\node at (2.8cm, -1.2cm){$m_i-\frac{\sqrt{n}}{2}$};
\node at (4.9cm, -2.5cm){Characters kept in $\mathcal{B}$};
\draw [decorate,decoration={brace,mirror,amplitude=10pt}](13cm,-1.6cm) -- (17.2cm,-1.6cm);
\node at (15.1cm, -2.5cm){Compare these characters with those kept in $\mathcal{B}$};
\end{tikzpicture}
\caption{The second pass allows us to find the longest $d$-near-palindrome by explicitly comparing characters.}\figlab{fig:recoverb}
\end{figure*}
We describe the second algorithm below, again omitting the case for when the longest $d$-near-palindrome has length at most $2\sqrt{n}$ and is therefore immediately recognized by the sliding window in the first pass. 
Recall that $\mathcal{C}$ has already been pruned in the first pass to only include checkpoints serving as the start of $d$-near-palindromes of length at least $\tilde{\ell}-\sqrt{n}$. 
We further prune $\mathcal{C}$ by removing checkpoints causing overlapping $d$-near-palindromes.

\begin{mdframed}
Preprocessing:
\begin{itemize}
\item[]
For each $c\in\mathcal{C}$:
\begin{itemize}
\item[]
If $c'<c<c'+\tilde{\ell}-\sqrt{n}$ for some other $c'\in\mathcal{C}$, then remove $c$. 
\end{itemize}
\end{itemize}
Second pass:
\begin{enumerate}
\item
Maintain a sliding window of size $2\sqrt{n}$ and set $\ell=\tilde{\ell}$ from the first pass. 
\item
Initialize $\mathcal{A}$ to be an empty array of size $\sqrt{n}$. It will dynamically contain the $\frac{\sqrt{n}}{2}$ characters before $c\in\mathcal{C}$ reported in the first pass.
\item
Initialize $\mathcal{B}$ to be an empty array of size $d\sqrt{n}$. It will dynamically contain the $\frac{\sqrt{n}}{2}$ characters before each of the at most $d$ different substrings between midpoints in each $L_c$.
\item
If $x=c-\frac{\sqrt{n}}{2}-j$ for some $c\in\mathcal{C}$ and $1\le j\le\frac{\sqrt{n}}{2}$, insert $S[x]$ into $\mathcal{A}$.
\item
If $x=m_i-\frac{\sqrt{n}}{2}-j$ for some $c\in\mathcal{C}, m_i\in L_c$ which has not been recorded, and $1\le j\le\frac{\sqrt{n}}{2}$, insert $S[x]$ into $\mathcal{B}$.
\item
If there exists $m_i,m_j\in L_c$ for some $c\in\mathcal{C}$ such that $x=m_j$ and $x-m_i\ge\ell-\sqrt{n}$, then call procedure $\recover(m_i,x,L_c)$ is see whether $S[m_i,x]$ is a $d$-near-palindrome.
\begin{itemize}
\item[]
If $S[m_i,x]$ is a $d$-near-palindrome, allocate space for $mismatches((m_i+x)/2)$ and set it to be the number of mismatches in $S[m_i,x]$. Also, keep the indices of the mismatches returned by procedure $\recover(m_i,x,L_c)$.
\end{itemize}
\item
If there exists $m\in L_c$ for some $c\in\mathcal{C}$ such that $m+\frac{\ell}{2}<x<m+\frac{\ell}{2}+\frac{\sqrt{n}}{2}$ and $S[x]\neq S[m-(x-m)+1]$ (which is stored in $\mathcal{A}$):
\begin{enumerate}
\item
If $mismatches(m)<d$, then insert $x$ into the set of mismatches and increment $mismatches(m)$.
\item
If $mismatches(m)=d$:
\begin{enumerate}
\item
If $2x-2m-2>\tilde{\ell}$, set $\tilde{\ell}=2x-2m-2$, $start=2m-x+1$, and $\mathcal{M}$ to be the set of mismatches.
\item
Deallocate space for $mismatches(m)$.
\item
If $m$ is the largest midpoint in $L_c$, remove the $\frac{\sqrt{n}}{2}$ characters in $\mathcal{A}$ before $c$.
\end{enumerate}
\end{enumerate}
\item
If there exists $m_i,m_j\in L_c$ for some $c\in\mathcal{C}$ such that $m_j-m_i\ge\ell-\sqrt{n}$, $0<x-m_j<\frac{\sqrt{n}}{2}$, $mismatches((m_i+m_j)/2)\le d$ and $S[x]\neq S[m_i-(x-m_j)+1]$ (which is stored in $\mathcal{B}$):
\begin{enumerate}
\item
If $mismatches((m_i+m_j)/2)<d$, increment $mismatches((m_i+m_j)/2)$, and insert $x$ into the set of mismatches.
\item
If $mismatches((m_i+m_j)/2)=d$:
\begin{enumerate}
\item
If $2x-m_i-m_j>\tilde{\ell}$, set $\tilde{\ell}=2x-m_i-m_j$, $start=m_i+m_j-x$, and $\mathcal{M}$ to be the set of mismatches.
\item
Deallocate space for $mismatches((m_i+m_j)/2)$.
\item
If $x>c+\ell+\sqrt{n}$, remove the characters in $\mathcal{B}$ associated with $c$.
\end{enumerate}
\end{enumerate}
\item
If $x=n$, then output $\tilde{\ell}$, $start$, and $\mathcal{M}$.
\end{enumerate}
\end{mdframed}

The correctness of the algorithm follows from the first pass recognizing the longest $d$-near-palindrome, possibly with the exception of up to $\frac{\sqrt{n}}{2}$ characters before checkpoints (and the corresponding $\frac{\sqrt{n}}{2}$ characters at the end of the $d$-near-palindrome. 
The first pass can keep this information by storing at most $\bigO{d}$ fingerprints, by \lemref{lem:periodic}. 
Then, procedure $\recover$ can fully recover the mismatches found in the first pass by reconstructing the fingerprints. 
Since the second pass dynamically keeps the $\frac{\sqrt{n}}{2}$ characters before checkpoints and candidate midpoints, then the remaining (at most $\sqrt{n}$) characters of the longest $d$-near-palindrome are explicitly checked and recognized. 
Therefore, the second pass returns exactly the longest $d$-near-palindrome.

\begin{lemma}
The total space used by the algorithm is $\bigO{d^2\sqrt{n}\log^6 n}$ bits. The update time per arriving symbol is $\bigO{d^2\sqrt{n}\log^5 n}$.
\end{lemma}
\begin{proof}
Each of the Karp-Rabin fingerprints consist of $\bigO{d\log^5 n}$ integers modulo $P$. 
Since $P\in[n^5,n^6]$, then $\bigO{d\log^6 n}$ bits of space are necessary for each fingerprint. 
There are $\sqrt{n}$ checkpoints, each of which may require $d$ fingerprints due to the compression allowed by the structural result. 
Hence, the space used by the fingerprints across all checkpoints is $\bigO{d^2\sqrt{n}\log^6 n}$ bits. 
Note that the algorithm also keeps $2\sqrt{n}$ characters in $\mathcal{A}$ and $d\sqrt{n}$ characters in $\mathcal{B}$, so the space usage by the algorithm follows.

For each arriving symbol $S[x]$, the algorithm checks each checkpoint and possibly the fingerprints of each checkpoint to check whether the substring is a $d$-near-palindrome. 
Since there are $\sqrt{n}$ checkpoints, each containing up to $d$ fingerprints of size $\bigO{d\log^5 n}$, and each subpattern of a fingerprint may be compared in constant time, then the overall update time is $\bigO{d^2\sqrt{n}\log^5 n}$.
\end{proof}
\section{Lower Bounds}
\seclab{sec:lb}
\begin{remindertheorem}{\thmref{thm:lowerbounds}}
\lowerbounds
\end{remindertheorem}
\begin{proofof}{\thmref{thm:lowerbounds}}
By Yao's Minimax Principle \cite{Yao77}, to show a $\Omega(d\log n)$ lower bound for randomized algorithms, it suffices to show a distribution over inputs such that every deterministic algorithm using less than $\frac{d\log n}{3}$ bits of memory fails with probability at least $\frac{1}{n}$.

We use an approach similar to \cite{GawrychowskiMSU16} who showed lower bounds for palindromes. 
Let $X$ be the set of binary strings of length $\frac{n}{4}$ with $d$ many $1$'s. 
Given $x\in X$, let $Y_x$ be the set of binary strings of length $\frac{n}{4}$ with either $\HAM(x,y)=d$ or $\HAM(x,y)=d+1$. 
We pick $(x,y)$ uniformly at random from $(X,Y_x)$. 

\begin{lemma}
\lemlab{lem:lb:mem}
Given an input $x\circ y$, any deterministic algorithm $\mathcal{D}$ which uses less than $\frac{d\log n}{3}$ bits of memory cannot correctly output whether $\HAM(x,y)=d$ or $\HAM(x,y)=d+1$ with probability at least $1-\frac{1}{n}$.
\end{lemma}
\begin{proof}
Note that $|X|=\binom{n/4}{d}$.
By Stirling's approximation, $|X|\ge\left(\frac{n}{4d}\right)^d$. 
Since $d=o(\sqrt{n})$, then $|X|\ge\left(\frac{n}{16}\right)^{d/2}$.

Because $\mathcal{D}$ uses less than $\frac{d\log n}{3}$ bits of memory, then $\mathcal{D}$ has at most $2^{\frac{d\log n}{3}}=n^{d/3}$ unique memory configurations. 
Since $|X|\ge\left(\frac{n}{16}\right)^{d/2}$, then there are at least $\frac{1}{2}(|X|-n^{d/3})\ge\frac{|X|}{4}$ pairs $x,x'$ such that $\mathcal{D}$ has the same configuration after reading $x$ and $x'$.
We show that $\mathcal{D}$ errs on a significant fraction of these pairs $x,x'$.

Let $\mathcal{I}$ be the positions where either $x$ or $x'$ take value $1$, so that $d+1\le|\mathcal{I}|\le 2d$.
Observe that if $\HAM(x,y)=d$, but $x$ and $y$ do not differ in any positions of $\mathcal{I}$, then $\HAM(x',y)>d$. 
Recall that $\mathcal{D}$ has the same configuration after reading $x$ and $x'$, so then $\mathcal{D}$ has the same configuration after reading $s(x,y)$ and $s(x',y)$.
But since $\HAM(x,y)=d$ and $\HAM(x',y)>d$, then the output of $\mathcal{D}$ is incorrect for either $s(x,y)$ or $s(x',y)$.

For each pair $(x,x')$, there are $\binom{n/4-|\mathcal{I}|}{d}\ge\binom{n/4-2d}{d}$ such $y$ with $\HAM(x,y)=d$, but $x$ and $y$ do not differ in any positions of $\mathcal{I}$. 
Hence, there are $\frac{|X|}{4}\binom{n/4-2d}{d}$ strings $s(x,y)$ for which $\mathcal{D}$ errs. 
We note that there is no overcounting because the output of $\mathcal{D}$ can be correct for at most one $s(x_i,y)$ for all $x_i$ mapped to the same configuration. 
Recall that $y$ satisfies either $\HAM(x,y)=d$ or $\HAM(x,y)=d+1$ so that there are $|X|\left(\binom{n/4}{d}+\binom{n/4}{d+1}\right)$ strings $s(x,y)$ in total. 
Thus, the probability of error is at least
\begin{align*}
\frac{\frac{|X|}{4}\binom{n/4-2d}{d}}{|X|\left(\binom{n/4}{d}+\binom{n/4}{d+1}\right)}
&=\frac{1}{4} \cdot \frac{\binom{n/4-2d}{d}}{\binom{n/4+1}{d+1}}
=\frac{(d+1)}{4}\frac{(n/4-3d+1)\ldots (n/4-2d)}{(n/4-d+1) \ldots  (n/4+1)}\\
&\ge\frac{d+1}{n+4}\left(\frac{n/4-3d+1}{n/4-d+1}\right)^d =\frac{d+1}{n+4}\left(1-\frac{2d}{n/4-d+1}\right)^d \\
&\ge\frac{d+1}{n+4}\left(1-\frac{2d^2}{n/4-d+1}\right)\ge\frac{1}{n}
\end{align*}
where the last line holds for large $n$, from Bernoulli's Inequality and $d=o(\sqrt{n})$.
\end{proof}
Define an infinite string $1^10^11^20^21^30^3\ldots$, and let $\nu$ be the prefix of length $\frac{n}{4}$. 
Given $x$ and $y$ from the above distribution, define string $s(x,y)=\nu^R xy^R\nu$ so that $s(x,y)$ is a $d$-near-palindrome of length $n$ if $\HAM(x,y)\le d$. 
\newcommand{\lemseparate}{If $\HAM(x,y)\ge d+1$, then the longest $d$-near-palindrome of $s(x,y)$ has length at most $200d^2+\frac{n}{2}$.
}
\begin{lemma} \lemlab{lem:separate}
\lemseparate
\end{lemma}
\begin{proof}
Suppose, by way of contradiction, that the longest $d$-near-palindrome of $s(x,y)$ has length at least $200d^2+\frac{n}{2}$. 
Since $\nu$ has length $\frac{n}{4}$ and $\HAM(x,y)>d$, then the midpoint $m$ of the longest $d$-near-palindrome of $s(x,y)$ lies within $x$ or $y$. 
Suppose that the midpoint is in $x$, so that $m<\frac{n}{2}$. 
We consider the cases where $m<\frac{n}{2}-8d$ and $m\ge\frac{n}{2}-8d$.

If $m<\frac{n}{2}-8d$, then at least $8d$ characters of $\nu$ coincide with characters of $xy$ in the reverse. 
However, the final $8d$ characters of $\nu$ contain at least $4d$ many $1$'s while the characters of $xy$ contain at most $2d+1$ many $1$'s, and so the Hamming distance is at least $2d-1$, which is a contradiction. 
On the other hand, if $m\ge\frac{n}{2}-8d$, then at least $200d^2$ characters of $\nu$ and $\nu^R$ coincide.
But because $m<\frac{n}{2}$, then the midpoint is closer to the end of $\nu^R$ than the beginning of $\nu$. 
Hence, for $k>8d$, each consecutive run of $k$ many $1$'s in $\nu^R$ corresponds with a $0$ in $\nu$. 
But then by the time $\nu^R$ has a consecutive run of $10d$ many $1$'s, the Hamming distance is at least $2d-1$, which is a contradiction. 
Since $\nu^R$ has a consecutive run of $10d$ many $1$'s by the index $(10d)^2=100d^2$, then the longest $d$-near-palindrome has length at most $200d^2+\frac{n}{2}$.

A similar argument follows if $m\ge\frac{n}{2}-8d$, so that the midpoint is in $y$.
\end{proof}
Since $d=o(\sqrt{n})$, then any algorithm with approximation factor $(1+\eps)$ can distinguish whether the longest $d$-near-palindrome in $s(x,y)$ has length $n$ or at most $200d^2+\frac{n}{2}$, for large $n$ and small and constant $\eps$. 
In turn, this algorithm can distinguish between $\HAM(x,y)=d$ and $\HAM(x,y)>d$ by \lemref{lem:separate}.
However, by \lemref{lem:lb:mem}, any algorithm using less than $\frac{d\log n}{3}$ bits of memory cannot distinguish between $\HAM(x,y)=d$ and $\HAM(x,y)>d$ with probability at least $1-1/n$.
 
Therefore, $\Omega(d\log n)$ bits of memory are necessary to $(1+\eps)$-approximate the length of the longest $d$-near-palindrome with probability at least $1-\frac{1}{n}$.
\end{proofof}

\begin{remindertheorem}{\thmref{thm:lowerboundsb}}
\lowerboundsb
\end{remindertheorem}
\begin{proofof}{\thmref{thm:lowerboundsb}}
We use a similar strategy as in \thmref{thm:lowerbounds} and analyze deterministic algorithms using less than $\frac{dn}{12E}$ memory, on a special hard distribution of inputs.

For $n'>0$, which we pick shortly, let $X$ be the set of binary strings of length $\frac{n'}{2}$. 
Given $x\in X$, let $Y_x$ be the set of binary strings of length $\frac{n'}{2}$ with either $\HAM(x,y)=d$ or $\HAM(x,y)=d+1$. 
We pick $(x,y)$ uniformly at random from $(X,Y_x)$.

\begin{lemma}
\lemlab{lem:lb:memb}
Given an input $x\circ y$, any deterministic algorithm $\mathcal{D}$ which uses less than $\frac{n'}{4}$ bits of memory cannot correctly output whether $\HAM(x,y)\le d$ or $\HAM(x,y)>d+1$ with probability at least $1-\frac{1}{n'}$, for $d=o(\sqrt{n})$.
\end{lemma}
\begin{proof}
Because $\mathcal{D}$ uses less than $ \frac{n'}{4}$ bits of memory, then $\mathcal{D}$ has at most $2^{n'/4}$ unique memory configurations.
Since $|X|=2^{n'/2}$, then there are at least $\frac{1}{2}(|X|-2^{n'/4})\ge\frac{|X|}{4}$ pairs $x,x'$ such that $\mathcal{D}$ has the same configuration after reading $x$ and $x'$. 
We show that $\mathcal{D}$ errs on a significant fraction of these pairs $x,x'$.

Let $\mathcal{I}$ be the positions where $x$ and $x'$ differ, so that $\HAM(x,x')=|\mathcal{I}|>0$.
Consider $i\neq\frac{|\mathcal{I}|}{2}$, so that either $i>\frac{|\mathcal{I}|}{2}$ or $i<\frac{|\mathcal{I}|}{2}$.
If $i<\frac{|\mathcal{I}|}{2}$, let $y$ differ from $x$ in $i$ positions (where $i\le d$) of $\mathcal{I}$ and in $d-i$ positions outside of $\mathcal{I}$.
Then $\HAM(x,y)=d$, but $\HAM(x',y)>d$.  
Similarly, let $y'$ differ from $x'$ in the $i$ positions of $\mathcal{I}$ and $d-i$ positions outside of $\mathcal{I}$ so that $\HAM(x',y')=d$, but $\HAM(x,y')>d$. 
Hence, $\mathcal{D}$ errs on either $(x,y)$ or $(x',y')$. 

There are at least $\sum_{i=0}^{\frac{|\mathcal{I}|}{2}-1}\binom{|\mathcal{I}|}{i}\binom{n'/2-|\mathcal{I}|}{d-i}$ such $y$ for each pair $(x,x')$.

Similarly, if $i>\frac{|\mathcal{I}|}{2}$, let $y$ differ from $x$ in the $i$ positions (where $i\le d$) of $\mathcal{I}$ and in $d+1-i$ positions outside of $\mathcal{I}$. 
Then $\HAM(x,y)=d+1$, but $\HAM(x',y)\le d$. 
Similarly, let $y'$ differ from $x'$ in the $i$ positions of $\mathcal{I}$ and $d+1-i$ positions outside of $\mathcal{I}$ so that $\HAM(x', y')=d$, but $\HAM(x',y')>d$. 
Hence, $\mathcal{D}$ errs on either $(x,y)$ or $(x',y')$. 
There are at least $\sum_{\frac{|\mathcal{I}|}{2}+1}^{d}\binom{|\mathcal{I}|}{i}\binom{n'/2-|\mathcal{I}|}{d+1-i}$ such $y$ for each pair $(x,x')$.

The total number of such $y$ is therefore at least:
\small
\begin{align*}
\sum_{i=0}^{\frac{|\mathcal{I}|}{2}-1}\binom{|\mathcal{I}|}{i}\binom{n'/2-|\mathcal{I}|}{d-i}+\sum_{\frac{|\mathcal{I}|}{2}+1}^{d}\binom{|\mathcal{I}|}{i}\binom{n'/2-|\mathcal{I}|}{d+1-i}
&\ge\sum_{i=0}^{\frac{|\mathcal{I}|}{2}-1}\binom{|\mathcal{I}|}{i}\binom{n'/2-|\mathcal{I}|}{d-i}+\sum_{\frac{|\mathcal{I}|}{2}+1}^{d}\binom{|\mathcal{I}|}{i}\binom{n'/2-|\mathcal{I}|}{d-i}\\
&\ge\left(\sum_{i=0}^d\binom{|\mathcal{I}|}{i}\binom{n'/2-|\mathcal{I}|}{d-i}\right)-\binom{|\mathcal{I}|}{|\mathcal{I}|/2}\binom{n'/2-|\mathcal{I}|}{d-|\mathcal{I}|/2}.
\end{align*}
Applying Vandermonde's identity, the total number of such $y$ is at least $\binom{n'/2}{d}-\binom{|\mathcal{I}|}{|\mathcal{I}|/2}\binom{n'/2-|\mathcal{I}|}{d-|\mathcal{I}|/2}$. 
Recall that $\mathcal{I}$ is the number of indices in which $x$ and $x'$ differ, so $|\mathcal{I}|\ge1$. 
Thus, 
\begin{align*}
\binom{|\mathcal{I}|}{|\mathcal{I}|/2}\binom{n'/2-|\mathcal{I}|}{d-|\mathcal{I}|/2}\le\binom{2}{1}\binom{n'/2-2}{d-1}
=\binom{n'/2}{d}\frac{2d(n'/2-d)}{(n'/2)(n'/2-1)}\le\binom{n'/2}{d}\frac{1}{\sqrt{n'}},
\end{align*}
\normalsize
where the last inequality comes from $d=o(\sqrt{n})$.
Therefore, for each pair $(x,x')$ the total number of errors by $\mathcal{D}$, as it cannot distinguish between $(x,y)$ and $(x',y')$, is at least $\frac{1}{2}\binom{n/2}{d}\left(1-\frac{1}{\sqrt{n}}\right)\ge\frac{1}{4}\binom{n/2}{d}$.

Since there are $\frac{|X|}{4}$ pairs of $(x,x')$, then there are at least $\frac{|X|}{16}\binom{n/2}{d}$ pairs $(x,y)$ for which $\mathcal{D}$ errs. 
Recall that $y$ satisfies either $\HAM(x,y)=d$ or $\HAM(x,y)=d+1$ so that there are $|X|\left(\binom{n'/2}{d}+\binom{n'/2}{d+1}\right)$ pairs $(x,y)$ in total. 
Thus, the probability of error is at least
\[\frac{\frac{|X|}{16}\binom{n'/2}{d}}{|X|\left(\binom{n'/2}{d}+\binom{n'/2}{d+1}\right)}\ge\frac{\binom{n'/2}{d}}{16\left(\binom{n'/2}{d}+(d+1)\binom{n'/2}{d}\right)}\ge\frac{1}{32d}.\]
Therefore for $d=o(\sqrt{n'})$, $\mathcal{D}$ fails with probability at least $\frac{1}{n'}$.
\end{proof}
Given strings $x$ and $y$ from the above distribution, define string 
\[s(x,y)=\textbf{1}^Ex_1\textbf{1}^{\frac{E}{d}}x_2\ldots \textbf{1}^{\frac{E}{d}}x_{n'/2}y_{n'/2}\textbf{1}^{\frac{E}{d}}\ldots y_2\textbf{1}^{\frac{E}{d}}y_1\textbf{1}^E,\]
where $x_i$ represents the $i$\th character of $x$ and $\textbf{1}^\ell$ represents $\ell$ repetitions of $1$. 
Let $s(x,y)$ be an input to $\mathcal{D}$ so that $s(x,y)$ has length $n=\left(\frac{E}{d}+1\right)(n'-2)+2E+2\le\frac{3E}{d}n'$. 
Note if $\HAM(x,y)\le d$, then $s(x,y)$ is a $d$-near-palindrome of length $\left(\frac{E}{d}+1\right)(n'-2)+2E+2$. 
However, if $\HAM(x,y)>d$, then the longest $d$-near-palindrome of $s(x,y)$ has length at most $\left(\frac{E}{d}+1\right)(n'-2)$. 
Consequently, any algorithm with additive error $E$ can be run on $s(x,y)$ to distinguish between $\HAM(x,y)\le d$ and $\HAM(x,y)\ge d+1$.
However, by \lemref{lem:lb:memb}, any algorithm using less than $\frac{n'}{4}$ bits of memory cannot distinguish between $\HAM(x,y)=d$ and $\HAM(x,y)>d$ with probability at least $1-\frac{1}{n'}>1-\frac{1}{n}$. 
Since $\frac{dn}{12}\le\frac{n'}{4}$, then it follows that for $d=o(\sqrt{n})$ and $E>d$, any randomized streaming algorithm which returns an additive approximation of $E$ to the length of the longest $d$-near-palindrome, with probability at least $1-\frac{1}{n}$, uses $\Omega\left(\frac{dn}{E}\right)$ space.
\end{proofof}
\section*{Acknowledgments}
We would like to thank Funda Erg{\"{u}}n, Tatiana Kuznetsova, and Qin Zhang for helpful discussions and pointers. 
We would also like to thank anonymous reviewers for pointing out an error in the proof of \thmref{thm:lowerboundsb}.
\def\shortbib{0}
\bibliographystyle{alpha}
\bibliography{references}
\end{document}